\documentclass[pra,aps,floatfix,amsmath,superscriptaddress,twocolumn]{revtex4}

\usepackage{amssymb}
\usepackage{graphicx}
\usepackage{graphics}
\usepackage{amsmath}
\usepackage{amsthm}
\usepackage{color}

\usepackage{mathtools}
\usepackage{relsize}

\begin{document}

\theoremstyle{plain}
\newtheorem{theorem}{Theorem}

\newtheorem{lemma}[theorem]{Lemma}
\newtheorem{corollary}[theorem]{Corollary}
\newtheorem{conjecture}[theorem]{Conjecture}
\newtheorem{proposition}[theorem]{Proposition}
\newtheorem{observation}[theorem]{Observation}

\theoremstyle{remark}
\newtheorem*{remark}{Remark}

\newcommand{\bea}{\begin{eqnarray}}
\newcommand{\eea}{\end{eqnarray}}
\newcommand{\e}{\eta}
\newcommand{\an}{\textbf{a}}
\newcommand{\bn}{\textbf{b}}
\newcommand{\cn}{\textbf{c}}
\def\bi{\begin{itemize}}
\def\ei{\end{itemize}}
\def\bc{\begin{center}}
\def\ec{\end{center}}
\def\E{{\cal E}}

\newcommand{\C}{\mathbb{C}}
\def\R{\hbox{$\mit I$\kern-.6em$\mit R$}}
\def\N{\hbox{$\mit I$\kern-.6em$\mit N$}}
\def\Y{Y}
\def\bk#1{\langle #1 \rangle}
\def\mw#1{\left< #1\right>}


\def\be{\begin{equation}}
\def\ee{\end{equation}}
\def\ba{\begin{align}}
\def\ea{\end{align}}

\newcommand{\mC}{\mathcal{C}}
\newcommand{\mE}{\mathcal{E}}
\newcommand{\mZ}{\mathcal{Z}}
\newcommand{\mU}{\mathcal{U}}
\newcommand{\mV}{\mathcal{V}}
\newcommand{\mA}{\mathcal{A}}
\newcommand{\mF}{\mathcal{F}}
\newcommand{\mI}{\mathcal{I}}
\newcommand{\mH}{\mathcal{H}}
\newcommand{\mL}{\mathcal{L}}
\newcommand{\mM}{\mathcal{M}}
\newcommand{\mT}{\mathcal{T}}
\newcommand{\mN}{\mathcal{N}}
\newcommand{\eqdef}{\equiv}

\newcommand{\fm}{\mathcal{F}_{\bf{m}}}
\newcommand{\am}{\mathcal{A}^{\textbf{m}}}
\newcommand{\dm}{\mathcal{D}(\mathrm{H}_{{\bf m}})}
\newcommand{\lr}{\rangle\langle}
\newcommand{\la}{\langle}
\newcommand{\ra}{\rangle}
\newcommand{\tr}{{\rm Tr}}

\newcommand{\mc}[1]{\mathcal{#1}}
\newcommand{\mbf}[1]{\mathbf{#1}}
\newcommand{\mbb}[1]{\mathbb{#1}}
\newcommand{\mrm}[1]{\mathrm{#1}}

\newcommand{\bra}[1]{\langle #1|}
\newcommand{\ket}[1]{|#1\rangle}
\newcommand{\braket}[3]{\langle #1|#2|#3\rangle}
\newcommand{\ip}[2]{\langle #1|#2\rangle}
\newcommand{\op}[2]{|#1\rangle \langle #2|}

\newcommand{\mbN}{\mathbb{N}}

\newcommand{\one}{\mbox{$1 \hspace{-1.0mm}  {\bf l}$}}

\definecolor{old}{rgb}{.0,.5,.2}
\newcommand{\old}[1]{#1}

\newcommand{\review}[1]{{\color{red} #1}}

\title{Transformations among Pure Multipartite Entangled States via Local Operations Are Almost Never Possible}
\author{David Sauerwein} \email{david.sauerwein@uibk.ac.at}
\affiliation{Institute for Theoretical Physics, University of Innsbruck, 6020 Innsbruck, Austria}
\author{Nolan R. Wallach}\email{nwallach@ucsd.edu}
\affiliation{Department of Mathematics, University of California/San Diego,
        La Jolla, California 92093-0112, USA}
\author{Gilad Gour}\email{giladgour@gmail.com}
\affiliation{
Department of Mathematics and Statistics, and Institute for Quantum Science and Technology (IQST),
University of Calgary, Alberta, Canada T2N 1N4}
\author{Barbara Kraus}\email{barbara.kraus@uibk.ac.at}
\affiliation{Institute for Theoretical Physics, University of Innsbruck, 6020 Innsbruck, Austria}


\begin{abstract}
Local operations assisted by classical communication (LOCC) constitute the free operations in entanglement theory. Hence, the determination of LOCC transformations is crucial for the understanding of entanglement. We characterize here almost all LOCC transformations among pure multipartite multilevel states. Combined with the analogous results for qubit states shown by Gour \emph{et al.} [J. Math. Phys. 58, 092204 (2017)], this gives a characterization of almost all local transformations among multipartite pure states. We show that nontrivial LOCC transformations among generic, fully entangled, pure states are almost never possible. Thus, almost all multipartite states are isolated. They can neither be deterministically obtained from local-unitary-inequivalent (LU-inequivalent) states via local operations, nor can they be deterministically transformed to pure, fully entangled LU-inequivalent states. In order to derive this result, we prove a more general statement, namely, that, generically, a state possesses no nontrivial local symmetry. We discuss further consequences of this result for the characterization of optimal, probabilistic single copy and probabilistic multi-copy LOCC transformations and the characterization of LU-equivalence classes of multipartite pure states.
\end{abstract}

\maketitle

\section{Introduction}
Entanglement lies at the heart of quantum theory and is the essential resource for many striking applications of quantum information science \cite{NiCh00, reviews, RaBr01, GoThesis, secretsharing, giovanetti}.
The entanglement properties of multipartite states are, moreover, fundamental to important concepts in condensed matter physics \cite{ScPe10}.
This relevance of entanglement in various fields of science has motivated great research efforts to gain a better understanding of these intriguing quantum correlations.

Local operations assisted by classical communication (LOCC) play an essential role in the theoretical and experimental investigation of quantum correlations. Spatially separated parties who share some entangled state can  utilize it to accomplish a certain task, such as teleportation.  The parties are free to communicate classically with each other and to perform any quantum operation on their share of the system. To give an example, party 1 would perform a generalized measurement on his/her system and send the result to all other parties. Party 2 performs then, depending on the measurement outcome of party 1, a generalized measurement. The outcome is again sent to all parties, in particular to party 3, who applies a quantum operation, which depends on both previous outcomes, on his/her share of the system, etc.. Any protocol that can be realized in such a way is a LOCC protocol. This physically motivated scenario led to the definition of entanglement as a resource that cannot be increased via LOCC.
Stated differently, entanglement theory is a resource theory where the free operations are LOCC. In particular, if $\ket{\psi}$ can be transformed to $\ket{\phi}$ via LOCC, then $E(\ket{\psi}) \geq E(\ket{\phi})$ for any entanglement measure $E$. Therefore, studying all possible LOCC transformations among pure states also leads to a partial order of entanglement.

In the bipartite case, simple, necessary, and sufficient conditions for LOCC transformations among pure states were derived \cite{nielsen}. This is one of
the main reasons why bipartite (pure state) entanglement is so well understood, as those conditions resulted in an elegant framework that explains how bipartite entanglement can be characterized, quantified and manipulated \cite{reviews}. In particular, the optimal resource of entanglement, i.e. the maximally entangled state, could be identified. It is, up to normalization and local unitary (LU) operations (which do not alter the entanglement), the state $\sum_i \ket{ii}$. This state can be transformed into any other state in the Hilbert space via LOCC. Many applications within quantum information theory, such as teleportation, entanglement-based cryptography, or dense coding, utilize this state as a resource.

In spite of considerable progress \cite{reviews, multip}, an analogous characterization of multipartite LOCC transformations remains elusive. The reasons for that are manifold. Firstly, the study of multipartite entangled states is difficult, and often intractable, due to the exponential growth of the dimension of the Hilbert spaces. Secondly, multipartite LOCC is notoriously difficult to describe mathematically \cite{chitambar}. Thirdly, there exist multipartite entangled states, belonging to the same Hilbert space, that cannot even be interconverted via stochastic LOCC (SLOCC) \cite{slocc} and, thus, there is no universal unit of multipartite entanglement.

Apart from LOCC transformations, other, more tractable local operations were considered. Local unitary (LU) operations, which as mentioned before, do not alter the entanglement, have been investigated \cite{Kr10}. SLOCC transformations, which correspond to a single branch of a LOCC protocol, have been analyzed \cite{slocc}. Both relations define an equivalence relation. That is, two states are said to be in the same SLOCC class (LU class) if there exists a $g\in \widetilde{G}$ ($g \in \widetilde{K}$) that maps one state to the other, respectively. Here, and in the following, $\widetilde{G}$ ($\widetilde{K}$) denotes the set of local invertible (unitary) operators. 
Clearly, two fully entangled states, i.e. states whose single-subsystem reduced states have full rank, have to be in the same SLOCC class in case there exists a LOCC transformation mapping one into the other. That is, it must be possible to locally transform one state into the other with a nonvanishing probability in case the transformation can be done deterministically. Apart from LU and SLOCC, where a single local operator is considered, transformations involving more operators have been investigated, such as LOCC transformations using only finitely many rounds of classical communication \cite{Spee0} or separable operations (SEP) \cite{cohen}. Considering only finitely many rounds of classical communication in a LOCC protocol is practically motivated and leads to a simple characterization of (generic) states to which some other state can be transformed to via such a protocol. However, it has been shown that there exist transformations that can only be accomplished with LOCC if infinitely many rounds of communication are employed \cite{chitinfty}. SEP transformations are easier to deal with mathematically than LOCC. However, they lack a clear physical meaning as they strictly contain LOCC \cite{chitambar, sepnotlocc}. Any separable map $\Lambda_{SEP}$ can be written as $\Lambda_{SEP}(\cdot) = \sum_k M_k (\cdot) M_k^\dagger$, where the Kraus operators $M_k = M_k^{(1)} \otimes \ldots \otimes M_k^{(n)}$ are local and fulfill the completeness relation $\sum_k M_k^\dagger M_k = \one$. In Ref. \cite{GoWa11} necessary and sufficient conditions for the existence of a separable map transforming one pure state into another were presented. Clearly, any LOCC protocol as explained above corresponds to a separable map. However, not any separable map can be realized with local operations and classical communication \cite{sepnotlocc} and there exist even multipartite pure state transformations that can be achieved via SEP, but not via LOCC \cite{HeSp15}.

Thus, despite all these efforts and the challenges involved in characterizing and studying LOCC, the fundamental relevance of LOCC within entanglement theory makes its investigation inevitable in order to reach a deeper understanding of multipartite entanglement. Already, the identification of the analog of the maximally entangled bipartite state, the maximally entangled set (MES), requires the knowledge of possible LOCC transformations. This set of states, which was characterized for small system sizes \cite{dVSp13,HeSp15,SpdV16}, is the minimal set of states from which any other fully entangled state (within the same Hilbert space) can be obtained via LOCC. The investigation of LOCC transformations, in particular for arbitrary local dimensions, might also lead to new applications in many fields of science, e.g. new ways to use quantum networks, which now become an experimental reality, or new theoretical tools in condensed matter physics. \\

Instead of investigating particular LOCC transformations we follow a different approach, which is based on the theory of Lie groups and algebraic geometry (see also \cite{GoKr16}). This new viewpoint allows us to overcome many of the usual obstacles in multipartite entanglement theory described above. It enables us to characterize, rather unexpectedly, all LOCC (and SEP) transformations, i.e. all local transformations, among pure states of a full-measure subset of any $(n>3)$-d--level (qudit) system and certain tripartite qudit systems. We show that there exists no nontrivial LOCC transformation from or to any of the states within this full-measure set. We call a local transformation nontrivial if it cannot be achieved by applying LUs (which can of course always be applied). To be more precise, we show that a generic state $\ket{\psi}$ can be deterministically transformed to a fully entangled state $\ket{\phi}$ via LOCC (and even SEP) if and only if (iff) $\ket{\phi} = u_1 \otimes \ldots \otimes u_n \ket{\psi}$, where $u_i$ is unitary; that is, only if $\ket{\psi}$ and $\ket{\phi}$ are LU equivalent. As LU transformations are trivial LOCC transformations, almost all pure multiqudit states are isolated. That is, they can neither be deterministically obtained from other states via nontrivial LOCC nor can they be deterministically transformed via nontrivial LOCC to other fully entangled pure states. This also holds if transformations via the larger class of SEP are considered.

We derive this result by using the fact that the existence of local symmetries of a state is essential for it to be transformable via LOCC or SEP (see \cite{GoKr16, GoWa11} and Sec. II). The local symmetries of a $n$-partite state $\ket{\psi}$ are all local invertible operators $g = g_1 \otimes \ldots \otimes g_n \in \widetilde{G}$ such that $g\ket{\psi} = \ket{\psi}$. The set of all local symmetries of $\ket{\psi}$ is also referred to as its stabilizer. We prove that for the aforementioned Hilbert spaces there exists a full-measure set of states that possess no nontrivial symmetry. These results are a generalization of those presented in \cite{GoKr16}. 
Here, the following remark is in order. One might be tempted to believe that the stabilizer of most states is trivial whenever the number of complex variables $N_v$ (describing $g$) in the equation $g\ket{\psi} = \ket{\psi}$ is smaller than the number of equations $N_e$ (describing $\ket{\psi}$). However, this counting argument already fails in the case of four qubits, where $13 = N_v < N_e = 16$ but only a zero-measure subset of states has trivial stabilizer and almost all states have nontrivial symmetries \cite{GoWa11, SpdV16, Wa08}. 
Hence, a parameter counting argument does not suffice to show that the set of states with trivial stabilizer is of full measure. In fact, a rigorous proof of this fact is already very involved for the qubit case. In \cite{GoKr16} methods from algebraic geometry and the theory of Lie groups were used to shown that generic $(n>4)$-qubit ($d=2$) states only have trivial symmetries. However, a straightforward generalization beyond qubit states was impossible, as in \cite{GoKr16} special properties of the qubit case, for instance the existence of so-called homogeneous SL-invariant polynomials (SLIPs) of low degree, were utilized. Note that, due to these special properties of qubit states, it was unclear whether indeed a similar result holds for arbitrary dimensions. As the statement is not true for less than five qubits, it could furthermore have turned out that the number of parties for which almost all states have trivial stabilizer depends on the local dimension, i.e. that $n$ depends on $d$. We show here, however, that this is not the case by employing new tools from algebraic geometry.
Clearly, the investigation of higher local dimensions is central in quantum information processing, where for e.g. in quantum networks the parties have access to more than just a single qubit. Moreover, in tensor network states, which are employed for the investigation of condensed matter systems, the local dimension is often larger than two.

A direct consequence of this result is that the maximally entangled set (MES) \cite{dVSp13} is of full measure in systems of $n>3$ qudits (and certain tripartite systems). The intersection of states which are in the MES and are convertible, i.e. which can be transformed into some other (LU--inequivalent) state are of measure zero. These states are the most relevant ones regarding pure state transformations. Prominent examples of these states are the GHZ-state \cite{ghz} or more generally stabilizer states \cite{GoThesis}. Hence, the results presented here do not only identify the full--measure set of states which are isolated, but also indicate which states can be transformed. 

As generic LOCC transformations are impossible, it is crucial to determine the optimal probabilistic protocol to achieve these transformations. Given the result presented here, the simple expression for the corresponding optimal success probability presented in \cite{GoKr16} also holds for a generic state with arbitrary local dimensions. Moreover, we show that the fact that almost no state possesses a nontrivial local symmetry can be used to derive simple conditions for two SLOCC-equivalent states to be LU--equivalent. We also show that our result leads to new insights into scenarios in which LOCC transformations of more than one copy of a state are considered. In particular, a lower bound on the probability, with which $n$ copies of a state $\ket{\psi}$ can be transformed into $m$ copies of a state $\ket{\phi}$ can be derived. Remarkably, this bound holds for any pair of states $\ket{\psi},\ket{\phi}$, i.e. even those which are not generic, and arbitrary numbers of copies, $n$, $m$. Furthermore, it leads to a new lower bound on the optimal rate to convert asymptotically many copies of $\ket{\psi}$ into copies of $\ket{\phi}$ via LOCC.

The rest of this paper is organized as follows. In Sec. \ref{sec:main} we present the main result of the paper and emphasis the physical consequences thereof. In particular, we first state that there exists a full--measure set of states (of almost all Hilbert spaces with constant local dimension), with the property that the local stabilizer of any state in this set is trivial (Theorem 1). We then recap why local symmetries play such an important role in state transformations and that Theorem 1 implies that generically there is no state transformation possible via LOCC. After that, we present further consequences of Theorem 1 for the characterization of optimal probabilistic LOCC transformations, of LU-equivalence classes and for the determination of probabilistic multi-copy LOCC transformations, as mentioned above. \\
In Sec. \ref{sec:math} we present the mathematical methods used to prove that almost all multiqudit states have trivial stabilizer. In Sec. \ref{SecPre} we introduce our notation and briefly recap the results presented in \cite{GoKr16}, where qubit systems were considered. In Sec. \ref{sec:gentriv} we develop these methods further and employ new tools from the theory of Lie groups and algebraic geometry to show that whenever there exists a so-called critical state whose set of unitary local symmetries is trivial then the stabilizer of a generic multipartite state is trivial (Theorem \ref{thmain}). In Sec. \ref{sec:states} we present examples of $n$--qudit systems for all local dimensions ($d>2$) and any number of subsystems ($n>3$) of states which have these properties. In particular, we prove there that the stabilizer of these states is trivial. Combined with Theorem \ref{thmain} mentioned above this shows that the stabilizer of a generic state, i.e. of a full measure subset of states, of $n>4$ qubits and $n>3$ qudits is trivial. This result also holds for three qudits with local dimension $d=4,5,6$.
In Sec. \ref{sec:picture} we illustrate and discuss the picture of multipartite pure state transformations that emerges if we combine this work with previous findings on bipartite \cite{nielsen}, 3-qutrit \cite{HeSp15} and qubit systems \cite{dVSp13,GoKr16}. In Sec. \ref{sec:conclusion} we present our conclusions. 

\section{Main results and implications}
\label{sec:main}
Let us state here the main results of this article and elaborate on its consequences in the context of entanglement theory.

We consider pure states belonging to the Hilbert space $\mH_{n,d}\eqdef\otimes^{n}\mathbb{C}^{d}$, i.e. the Hilbert space of $n$ qudits. Whenever we do not need to be specific about the local dimensions, we simply write ${\cal H}_n$ instead of $\mH_{n,d}$. As before, $\widetilde{G}$ denotes the set of local invertible operators on $\mH_n$. Our main result concerns the group of local symmetries of a multipartite state $\ket{\psi}$, also referred to as its stabilizer in $\widetilde{G}$, which is defined as 
\bea
\widetilde{G}_\psi\eqdef \left\{g\in \widetilde{G} \;\Big|\;g|\psi\ra=|\psi\ra\right\}\subset \widetilde{G}.
\eea
We prove that for almost all multiqudit Hilbert spaces, $\mH_{n,d}$, there exists a full-measured set of states whose stabilizer is trivial. Recall that a subset of ${\cal H}_n$ is said to be of full measure, if its complement in ${\cal H}_n$ is of lower dimension. Stated differently, almost all states are in the full-measured set and its complement is a zero-measure set.

The main result presented here is given by the following theorem.

\begin{theorem}\label{IntroMain}
For any number of subsystems $n>3$ and any local dimension $d > 2$ there exists a set of states whose stabilizer in $\widetilde{G}$ is trivial. This set is open, dense and of full measure in $\mH_{n,d}$. Such a set of states also exists for $n=3$ and $d=4,5,6$.
\end{theorem}
Note that it will be clear from the proof of the theorem why the case $n=3$ has to be treated differently (see Sec. \ref{sec:math}). However, it is likely that the statement of the theorem also holds for $n=3$ and $d>6$. 
Theorem \ref{IntroMain} shows that almost all multiqudit states $\ket{\psi}$ have only the trivial local symmetry, i.e. $\widetilde{G}_\psi = \{\one\}$. This result has deep implications for entanglement theory.
In order to explain them, we briefly review the connection between the local symmetries of multipartite states and their transformation properties under LOCC and SEP.\\

As mentioned in the introduction, we say that $\ket{\psi}$ can be transformed via SEP into $\ket{\phi}$ if there exists a separable map $\Lambda_{SEP}(\cdot) = \sum_k M_k (\cdot) M_k^\dagger$ such that $\Lambda_{SEP}(\ket{\psi}\bra{\psi}) = \ket{\phi}\bra{\phi}$,  where the Kraus operators $M_k = M_k^{(1)} \otimes \ldots \otimes M_k^{(n)}$ are local and fulfill the completeness relation $\sum_k M_k^\dagger M_k = \one$.
The transformation is possible via LOCC if $\Lambda_{SEP}$ can be implemented locally. It is clear that a fully entangled state $\ket{\psi}$ can only be transformed into an other fully entangled state $\ket{\phi}$ if these states are SLOCC equivalent, i.e. $\ket{\phi} = h\ket{\psi}$ for some $h \in \widetilde{G}$. In \cite{GoWa11} it was shown that a fully entangled state $\ket{\psi}$ can be transformed via SEP to $\ket{\phi} = h\ket{\psi}$ iff there exists a $m \in \mathbb{N}$ and a set of probabilities
$\{p_k\}_{k=1}^m$ ($p_k\geq 0, \sum_{k=1}^m p_k=1$) and $\{S_k\}_{k=1}^m \subset
\tilde{G}_\psi$ such that 
\bea 
\sum_k p_k S_k^\dagger H S_k=r \one. \label{eq:SEP}
\eea 
Here, $H=h^\dagger h \equiv \bigotimes H_i$ is a local operator and
$r=\frac{||\ket{\phi}||^2}{||\ket{\psi}||^2}$.  This criterion for the existence of a SEP transformation can be understood as follows. Let $M_k$ denote the local operator which maps $\ket{\psi}$ to $\ket{\phi}= h \ket{\psi}$, i.e. $M_k \ket{\psi} = c_k h\ket{\psi}$ for some $c_k \neq 0$. Hence, $ h^{-1} M_k$ must be proportional to a local symmetry of $\ket{\psi}$. Using then the completeness relation $\sum_k M_k^\dagger M_k = \one$ leads to the necessary and sufficient conditions in Eq. (\ref{eq:SEP}) for the existence of a separable map transforming one fully entangled state into the other \cite{GoWa11}.

As LOCC is contained in SEP, it is evident from this result that the local symmetries of a state play also a major role in the study of LOCC transformations. However, in order to characterize LOCC transformations among fully-entangled states using Eq. (\ref{eq:SEP}) one has to determine their local symmetries, find all solutions of Eq. (\ref{eq:SEP}) and check if the corresponding separable measurement can be implemented locally. For particular pairs of states such a procedure is feasible, even though it might be very tidious. However, to find all possible LOCC transformations seems infeasible. Our main result (see Theorem \ref{IntroMain}) allows to accomplish all the steps described above for almost all multipartite qudit states and thereby provides a characterization of deterministic SEP and LOCC transformations for almost all qudit states. This is one of the reasons why Theorem \ref{IntroMain} has such deep implications in entanglement theory, as we explain below.

In \cite{GoKr16} some of us proved a similar result as stated in Theorem \ref{IntroMain} for qubit states.
There, so-called SL-invariant polynomials (SLIPS) \cite{GoWa13} were used to identify a full-measure subset of all $(n>4)$-qubit states that have trivial stabilizer. As the special characteristics of the qubit case, for instance the existence of SLIPS of low degree, cannot be utilized for higher dimensions, this proof does not hold beyond qubit states. Precisely due to these peculiarities of qubit states it was unclear whether indeed a similar result holds for arbitrary local dimensions. Moreover, as the analog of Theorem \ref{IntroMain} is not true for less than five qubits, it could have turned out that the number of parties for which almost all states have a trivial stabilizer depends on the local dimension, i.e. that $n$ depends on $d$. Theorem \ref{IntroMain} shows that this is not the case. In order to tackle the case of arbitrary local dimensions, we employ in this work new tools from the theory of Lie groups and geometric invariant theory without explicitly using SLIPs (see Sec. \ref{sec:math}). We also show in Sec. \ref{sec:math} that the new results encompass the qubit case.

Let us now discuss the consequences of Theorem \ref{IntroMain} in the context of entanglement theory.\\

\subsection{Nontrivial deterministic local transformations are almost never possible}
In \cite{GoKr16} it was shown that states with trivial stabilizer are isolated. That is, a state with trivial stabilizer can neither be obtained from LU-inequivalent states via LOCC, nor can it be transformed to LU-inequivalent fully entangled states via LOCC. The same holds for transformations via SEP. Indeed, for such a state the only solution to Eq. (\ref{eq:SEP}) is $H = \one$, which means that $\ket{\psi}$ is LU equivalent to $\ket{\phi}$. It was then shown in \cite{GoKr16} that this holds for almost all states of $n>4$ qubits. Theorem \ref{IntroMain} ensures that the same holds true for almost all multiqudit states, which is stated in the following theorem.

\begin{theorem}
\label{thm:trivtrans}
Let $\mH_{n,d}$ be one of the multipartite qudit Hilbert spaces in Theorem \ref{IntroMain} and let $\ket{\psi} \in \mH_{n,d}$ be a fully entangled $n$-partite state with trivial stabilizer, i.e.  $\widetilde{G}_\psi = \{\one\}$.
  Then, $\ket{\psi}$ can be deterministically obtained from or transformed to a fully entangled $\ket{\phi}$ via $LOCC$ or $SEP$ iff $\ket{\psi}$ and $\ket{\phi}$ are related by local unitary operations; that is, iff there exists a $u \in \tilde{K}$ such that $\ket{\psi} = u\ket{\phi}$.
\end{theorem}

Recall that $\tilde{K}$ denotes the group of local unitary operators. On the one hand, this result shows that, rather unexpectedly, a characterization of LOCC transformations of almost all multiqudit states is possible. On the other hand, it proves that these transformations are generically extremely restricted and nontrivial transformations are generically impossible. That is, the parties who share a generic state cannot transform it via LOCC deterministically into any other (LU--inequivalent) state. This result might also be the reason why there has been so little progress on multipartite state (or entanglement) transformations via local operations.\\

As the MES is defined as the minimal set of states which can be transformed into any other fully entangled state in the Hilbert space \cite{dVSp13}, Theorem \ref{thm:trivtrans}  implies that the MES of $(n>3)$-qudits is of full measure. Note that this is in strong contrast to the bipartite case, where a single state, namely the maximally entangled state $\ket{\Phi^+}=\sum_i \ket{ii}$ can be transformed into any other state in the Hilbert space with LOCC. In Sec. \ref{sec:picture} we discuss in detail the picture of multipartite pure state transformations that emerges if we combine our findings with previous results on the subject (see also Fig. \ref{fig:1}).\\

Theorem \ref{thm:trivtrans} also has implications for the construction of entanglement measures. Recall that an entanglement measure for pure states is a function $E: \mH_n \rightarrow \mathbb{R}_{\geq 0}$ such that $E(\psi) \geq E(\phi)$ holds whenever the transformation from $\ket{\psi}$ to  $\ket{\phi}$ can be performed deterministically via LOCC. Since generic multiqudit states cannot be reached via nontrival deterministic LOCC, one only has to verify if $E$ is invariant under LU-transformations and nonincreasing under LOCC transformations to and within the zero-measure subset of states with nontrivial stabilizer, e.g. to states that are not fully entangled.\\

\subsection{A characterization of optimal probabilistic local transformations for almost all multiqudit states}

Given the fact that it is not possible to transform generic multiqudit states via local transformations into any other state, it is crucial to determine the optimal probability to achieve these conversions. Note that if both, the initial and final states, are fully entangled, this probability is only nonzero if they are elements of the same stochastic LOCC (SLOCC) class \cite{slocc}. In \cite{GoKr16} some of us found an explicit formula for this probability for qubit--states. Due to Theorem \ref{IntroMain} this formula indeed holds for arbitrary local dimensions.
\begin{theorem}
\label{thm:opttrans}
Let $\mH_{n,d}$ be one of the multipartite qudit Hilbert spaces in Theorem \ref{IntroMain}, let $\ket{\psi} \in \mH_{n,d}$ be a normalized fully entangled $n$-partite state with trivial stabilizer, i.e.  $\widetilde{G}_\psi = \{\one\}$, and let $\ket{\phi} = h\ket{\psi}$ be a normalized state in the $SLOCC$ class of $\ket{\psi}$. Then the maximum probability to convert $\ket{\psi}$ to $\ket{\phi}$ via $LOCC$ or $SEP$ is given by
\begin{align}
 p_{max}(\ket{\psi}\rightarrow\ket{\phi}) = \frac{1}{\lambda_{max}(h^\dagger h)}, \label{eq:optprob}
\end{align}
where $\lambda_{max}(X)$ denotes the maximal eigenvalue of $X$.
\end{theorem}

Due to Theorem \ref{IntroMain} this theorem gives a simple expression for the optimal probability $p_{max}(\ket{\psi}\rightarrow \ket{\phi})$ to locally transform a generic $(n>3)$-qudit state $\ket{\psi}$ into another fully entangled state $\ket{\phi}$. These results also hold for tripartite $d$-level systems with $d=4,5,6$. It should be noted here that the optimal success probability was only known for very restricted transformations prior to these results (see e.g. \cite{turgut, GoWa11} and references therein). Theorem \ref{thm:trivtrans} and Theorem \ref{thm:opttrans} now provide a characterization of deterministic and optimal probabilisitc local transformations for almost all multiqudit states.\\

Note that the optimal success probability given in Eq. (\ref{eq:optprob}) is optimal for transformations via LOCC and via SEP. This shows that, despite the fact that there are pure state transformations that can be achieved via SEP but not via LOCC \cite{HeSp15}, the two classes of operations are equally powerful for transformations among generic $(n>3)$-qudit states. The reason for this is that the optimal SEP protocol is a so-called one-successful-branch protocol (OSBP), which can always be implemented via LOCC in one round of classical communication. As suggested by the name, a OSBP is a simple protocol for which only one measurement branch leads to the final state, while all other branches lead to states that are no longer fully entangled (see \cite{GoKr16}). This optimal protocol to transform $\ket{\psi}$ into $\ket{\phi} = h\ket{\psi}$ via LOCC, where $h = h_1 \otimes \ldots \otimes h_n \in \widetilde{G}$, is implemented as follows. The first party applies a local generalized measurement that contains an element proportional to $h_1$. Similarly, party 2 applies a local generalized measurement that contains an element proportional to $h_2$ etc. The successful branch is the one where all parties managed to apply the operator $h_i$. Due to the fact that the local measurements have to obey the completeness relation one can show that the maximal success probability is given as in Theorem \ref{thm:opttrans}. Note that this protocol can of course also be performed if the corresponding state has nontrivial symmetries. That is, the success probability given in Eq. (\ref{eq:optprob}) is always a lower bound on the success probability.

Due to Theorem \ref{thm:trivtrans}, the optimal success probability can only be one if the states are LU--equivalent. Let us verify that this is indeed the case.
Given the premises of Theorem \ref{thm:opttrans} we make the following observation.

\begin{observation} The optimal success probability as given in Theorem \ref{thm:opttrans} is equal to one iff $H\equiv h^\dagger h=\one$. \end{observation}

This can be easily seen as follows. As $\ket{\psi}$ and $\ket{\phi}=h\ket{\psi}$ are both normalized we have that
\bea \lambda_{max}(H)=max_{\chi} \frac{\bra{\chi}H\ket{\chi}}{\bra{\chi}\chi\rangle}\geq  \frac{\bra{\psi}H\ket{\psi}}{\bra{\psi}\psi\rangle}=1.\eea Due to Eq. (\ref{eq:optprob}) the success probability is one iff the maximal eigenvalue of $H$ is one. We hence obtain that $\ket{\psi}$ is an eigenstate of $H$, i.e. $H\ket{\psi}=\lambda_{max}(H) \ket{\psi}$. However, as $H$ is in $\widetilde{G}$ and as $\ket{\psi}$ does not have any nontrivial local symmetry it must hold that $H=\one$.

\subsection{A simple method to decide LU-equivalence of generic multiqudit states}

Since local unitary transformations are the only trivial LOCC transformations of pure states, i.e. the only transformations that do not change the entanglement of a state \cite{gingrich}, it is important to know when two states are LU-equivalent. That is, given two states $\ket{\psi},\ket{\phi}$ one would like to know whether there exists a local unitary $u \in \widetilde{K}$ such that $\ket{\psi} = u\ket{\phi}$. In general this is a highly nontrivial problem (see e.g. \cite{Kr10}).  However, we show now that the results in this article also allow us to solve the LU-equivalence problem for generic multiqudit states, as stated in the following theorem.

\begin{theorem}
\label{thm:LU1}
Let $\ket{\psi}, \ket{\phi} \in \mH_n$ be both states in the SLOCC class of a state, $\ket{\psi_s}$, with trivial stabilizer, i.e. $\widetilde{G}_{{\psi_s}} = \{\one\}$. That is, $\ket{\psi} = g\ket{\psi_s}$ and $\ket{\phi} = h\ket{\psi_s}$. Then $\ket{\psi}$ is LU--equivalent to $\ket{\phi}$ iff $G= H$. \end{theorem}

\begin{proof}
As before we use the notation $G=g^\dagger g$ and $H= h^\dagger h$. First, note that $G=H$ holds iff $g=u h$ for some local unitary $u \in \widetilde{K}$. Hence, $\ket{\psi} = g\ket{\psi_s}=u h \ket{\psi_s}=u  \ket{\phi}$ and therefore the states are LU--equivalent. The other direction of the proof can be seen as follows. If $\ket{\psi} = g\ket{\psi_s}=u h \ket{\psi_s}$, then $g^{-1}uh=\one $ must hold, as $\ket{\psi_s}$ does not possess any nontrivial local symmetry. Thus, we have that $G=H$.
\end{proof}

This strong implication follows only from the fact that $\ket{\psi_s}$ has trivial stabilizer, which implies that the standard form $g \ket{\psi_s}$ with which a state in the SLOCC class of $\ket{\psi_s}$ can be represented, is unique. That is, the only $g'$ such that $g \ket{\psi_s}=g' \ket{\psi_s}$ is $g = g'$, as otherwise $(g')^{-1}g$ would be a nontrivial local symmetry of the state $\ket{\psi_s}$. 
Due to Theorem \ref{IntroMain}, Theorem \ref{thm:LU1} applies to almost all multiqudit states.

Let us now generalize this result to the situation where it is known that the two states are in the same SLOCC class, but the local invertible operator transforming one into the other (for the states above the operator $h g^{-1}$) is unknown. To this end, we introduce now the notion of critical states. A state is called critical if all of its single-subsystem reduced states are proportional to the completely mixed state \cite{GoWa11}. Prominent examples of critical states are Bell states, GHZ states~\cite{ghz}, cluster states~\cite{RaBr01}, graph states~\cite{Hei05}, code states~\cite{NiCh00}, and absolutely maximally entangled states \cite{helwig}. The set of all critical states in $\mH_{n,d}$, denoted by $Crit(\mH_{n,d})$, plays an important role in entanglement theory as the union of all SLOCC classes of critical states is of full measure in $\mH_{n,d}$ \cite{GoWa11}. For more details and properties of critical states we refer the reader to Sec. \ref{sec:math}.

Let us note that the standard form, $\ket{\psi} = g\ket{\psi_s}$, of a generic state, corresponds to the normal form introduced in \cite{NFFrank}. The numerical algorithm presented in \cite{NFFrank} can be used to find the normal form of a generic state, i.e. a local invertible $g \in \widetilde{G}$ and a critical state $\ket{\psi_s}$ such that $\ket{\psi}=g \ket{\psi_s}$. Due to the Kempf-Ness theorem (\cite{KN}, see also Appendix \ref{app:theorems}), there exists, up to local unitaries, only one critical state in a SLOCC class. Hence, computing the normal form for two states in the same SLOCC class leads to $\ket{\psi}=g \ket{\psi_s}$ and $\ket{\phi} = h\ket{\psi'_s}$, where $\ket{\psi'_s}=u\ket{\psi_s}$, with $u$ a local unitary. The question we address next is when are these two states LU--equivalent. The necessary and sufficient condition is given by the following lemma.

\begin{lemma}

Let $\ket{\psi}, \ket{\phi} \in \mH_n$ be both states in the SLOCC class of a critical state, $\ket{\psi_s}$, with trivial stabilizer. Let further $\ket{\psi} = g\ket{\psi_s}$ and $\ket{\phi} = h\ket{\psi'_s}$ be the normal forms of the states derived with the algorithm presented in \cite{NFFrank}. Then $\ket{\psi}$ is LU--equivalent to $\ket{\phi}$ iff the local unitary $u$ which transforms $\ket{\psi'_s}$ into $\ket{\psi_s}$, i.e. $\ket{\psi_s}=u \ket{\psi'_s}$ (which must exist and is unique) fulfills $G=u^\dagger H u$.
\end{lemma}
Due to Theorem \ref{IntroMain} this theorem again applies to almost all multiqudit states. It provides an easy way to solve the, a priori highly nontrivial, problem of deciding LU equivalence of two generic states that are SLOCC equivalent.
\begin{proof}
{\it If}: Let $G=u^\dagger H u$. Then there exists a unitary $v$ such that $g=v h u $. As all operators, $u$, $h$, and $g$ are local and invertible, $v$ is a local unitary operator. Hence, $\ket{\psi} = g\ket{\psi_s}= vh u \ket{\psi_s}=vh \ket{\psi'_s} = v\ket{\phi}$. {\it Only if}: If there exists a local unitary, $v$ transforming $\ket{\phi}$ into $\ket{\psi}$ we have $\ket{\psi}=g\ket{\psi_s}= v
\ket{\phi}=v h\ket{\psi'_s}=vh u\ket{\psi_s}$. The last equality follows from the uniqueness of the critical states in a SLOCC orbit. As $\ket{\psi_s}$ does not possess any non--trivial local symmetry it must hold that $g=vh u$.  Therefore, we have $G=u^\dagger H u$.

\end{proof}

\subsection{Multi-copy transformations and asymptotic conversion rates}

Let us briefly discuss which consequences the results presented here have in the case where transformations of many copies of a state are considered. First of all, note that the fact that $\ket{\psi}$ has only trivial local symmetries does not imply that the same holds for multiple copies of this state. In fact, any  $k$ copies of a state $\ket{\psi}$, i.e.  $\ket{\psi}^{\otimes k}$, do have local symmetries, namely a local permutation operator (SWAP) applied to all parties. Hence, multi-copy states belong to the zero-measure subset of multiqudit states with nontrivial local symmetries. These nontrivial symmetries could give rise to nontrivial local transformations (see Eq. (\ref{eq:SEP})). Indeed, it has recently been shown in \cite{CrHeKr18} that there are cases where two copies of a state can be transformed to states which cannot be reached from other states in the case of a single copy. Hence, the MES can be made smaller even if only two copies of the state are considered.\\
However, since we know the optimal probability to locally transform a single copy of a generic state $\ket{\psi}$ into a fully entangled state $\ket{\phi}$, it is straighforward to obtain a lower bound on the optimal probability to transform $k$ copies of $\ket{\psi}$ into $m \leq k$ copies of a fully entangled state $\ket{\phi} = h\ket{\psi}$ via LOCC, namely
\begin{align}
&p_{max}(\ket{\psi}^{\otimes k} \rightarrow \ket{\phi}^{\otimes m})  \nonumber \geq\\ 
&\sum_{j=m}^k{{k}\choose{j}} p_{max}(\ket{\psi}\rightarrow \ket{\phi})^{j} (1-p_{max}(\ket{\psi}\rightarrow \ket{\phi}))^{k-j} \label{eq:ktom}
\end{align}
Although this bound follows trivially from our results on single copy transformations it can provide new insights into the multi-copy case. This is examplified if one considers the asymptotic limit of $k\rightarrow \infty$, where one is interested in the optimal rate $R(\ket{\psi}\rightarrow \ket{\phi})$ at which asymptotically many copies of a state $\ket{\psi}$ can be transformed into copies of a state $\ket{\phi}$, which is defined as
\begin{widetext}
\begin{align}
 &R(\ket{\psi}\rightarrow \ket{\phi}) = \sup\left\{r \ \vert \ \lim_{k\rightarrow \infty} \left ( \inf_{\Lambda_{LOCC}}\| \Lambda_{LOCC}(\ket{\psi}\bra{\psi}^{\otimes k}) - \ket{\phi}\bra{\phi}^{\otimes \lfloor rk \rfloor}\|_1\right) = 0 \right\}. \label{eq:ratedef}
\end{align}
\end{widetext}
Here, the infimum is taken over all LOCC maps and $\|X\|_1 =  \tr(\sqrt{X^\dagger X})$ denotes the trace norm of $X$.
It was recently shown in \cite{eisert} that for tripartite states $\ket{\psi}, \ket{\phi}$ it holds that
\begin{align}
 R(\ket{\psi} \rightarrow \ket{\phi}) \geq \min\left\{\frac{S(\rho_{\psi}^{(1)})}{S(\rho_{\phi}^{(2)}) + S(\rho_{\phi}^{(3)})},\frac{S(\rho_{\psi}^{(2)})}{S(\rho_{\phi}^{(2)})},\frac{S(\rho_{\psi}^{(3)})}{S(\rho_{\phi}^{(3)})}\right\}, \label{eq:eisert}
\end{align}
where $\rho_{\psi}^{(i)} = \tr_{l \neq i}(\ket{\psi}\bra{\psi})$ (and similar for $\ket{\phi}$) and $S(\rho) = -\tr[\rho \log(\rho)]$ is the Von Neumann entropy. Note that this bound can obviously be improved by taking the maximum over all bipartitions of the tripartite states. Little is known on lower bounds on $R(\ket{\psi}\rightarrow \ket{\phi})$ for states of more than three parties. 
However, due to Ineq. (\ref{eq:ktom}) and the law of large numbers (see e.g. \cite{NiCh00} and references therein) we obtain the following theorem.
\begin{theorem}
\label{thm:manycopy}
Let $\ket{\psi},\ket{\phi} \in \mH_n$ be two multipartite entangled states and let $p_{max}(\ket{\psi}\rightarrow \ket{\phi})$ denote the optimal success probability to transform $\ket{\psi}$ into $\ket{\phi}$ via LOCC. Then the asymptotic LOCC conversion rate from $\ket{\psi}$ to $\ket{\phi}$ fulfills,
\begin{align}
 R(\ket{\psi}\rightarrow \ket{\phi}) \geq p_{max}(\ket{\psi}\rightarrow \ket{\phi}). \label{eq:multirate}
\end{align}
\end{theorem}
For a normalized generic multiqudit state $\ket{\psi}$ (i.e. with trivial stabilizer) and a normalized state $\ket{\phi} = h\ket{\psi}$ we can insert the expression of Eq. (\ref{eq:optprob}) for $p_{max}(\ket{\psi}\rightarrow \ket{\phi})$ into Ineq. (\ref{eq:multirate}) and we obtain the following bound,
\begin{align}
R(\ket{\psi} \rightarrow \ket{\phi}) \geq \frac{1}{\lambda_{max}(H)}. \label{eq:asymp}
\end{align}
Note that, even in the tripartite case (e.g. for three 4-level systems), one can easily construct examples where the bound in Ineq. (\ref{eq:asymp}) is better than the bound in Ineq. (\ref{eq:eisert}) (even if optimized over all bipartitions), while there are also tripartite states for which the opposite holds.\\

\section{Mathematical concepts and proof of the main result}
\label{sec:math}
In this section we present the proof of our main result, Theorem \ref{IntroMain}. In fact, we prove Theorem \ref{IntroMain} by deriving results that are stronger than actually required. However, we believe that these tools are also useful in other contexts and should therefore be presented in the main text of this article. We first introduce in Sec. \ref{SecPre} our notation and the main mathematical tools that we use. Furthermore, we summarize some of the results which were presented in \cite{GoKr16}. In Sec. \ref{sec:gentriv} we first give a concise outline of the proof of Theorem \ref{IntroMain}. We then continue with a presentation of the detailed proof. In Sec. \ref{sec:states} we give examples of states with trivial stabilizer, which are required to complete the proof.

\subsection{Notations and preliminaries}
\label{SecPre}

Throughout the remainder of this paper we use the following notation. We consider the following 4 different groups all acting
on $\mH_n$:
\begin{align*}
& G\eqdef SL(d,\mathbb{C})\otimes\cdots\otimes
SL(d,\mathbb{C})\subset SL(\mH_{n})\\
& K\eqdef SU(d)\otimes\cdots\otimes
SU(d)\subset SU(\mH_{n})\\
& \tilde{G}\eqdef GL(d,\mathbb{C})\otimes\cdots\otimes
GL(d,\mathbb{C})\subset GL(\mH_{n})\\
& \tilde{K}\eqdef U(d)\otimes\cdots\otimes
U(d)\subset U(\mH_{n})
\end{align*}
Note that $\widetilde{G}=\mathbb{C}^{\times}G$, where $\mathbb{C}^{\times} =\mathbb{C}\backslash\{0\}$ and that $\widetilde{K}=\widetilde{G}\cap U(n)$. That is
$\widetilde{K}=\{zu|u\in K,\left\vert z\right\vert =1\}$.

Given a subgroup $H \subset GL(\mH_{n})$,  the \emph{stabilizer} subgroup of a state $|\psi\ra\in\mH_n$ with respect to this group is defined as
\begin{equation*}
H_\psi\eqdef \left\{h\in H\;\Big|\;h|\psi\ra=|\psi\ra\right\}\subset H.
\end{equation*}
If we refer to the stabilizer of a state $\ket{\psi}$ without explicitly mentioning the corresponding group, we mean $\widetilde{G}_\psi$.
Moreover, the orbit of a state $\ket{\psi}$ under the action of $H$ is defined as
\begin{equation*}
H|\psi\rangle\eqdef\left\{h|\psi\rangle\;\Big|\;h\in H\right\}.
\end{equation*}
Note that the orbit contains states that are not necessarily normalized,
and any orbit $H|\psi\ra$ is an embedded submanifold of $\mH_n$. Hence, any orbit $H\ket{\psi}$ has a dimension, which we denote by $\dim (H\ket{\psi})$.

In Sec. \ref{sec:main} we briefly mentioned the set of critical states, $Crit(\mathcal{H}_{n})$, in $\mathcal{H}_n$ which contains all states whose single-subsystem reduced states are proportional to the completely mixed state. Denoting by $Lie(G)$ the Lie algebra of $G$, this set can also be expressed as
\bea
\label{eq:critstates}
Crit(\mathcal{H}_{n})\equiv\{\ket{\phi}
\in\mathcal{H}_{n}|\left\langle \phi|X|\phi\right\rangle =0, \forall X\in
Lie(G)\}.
\eea
Criteria for when a system with Hilbert space $\mH_n$ contains critical states were found in \cite{bryan1, bryan2}. If $Crit(\mathcal{H}_{n})$ is not empty the union of all orbits (in $G$) containing a critical state, i.e. $G\cdot Crit(\mH_n)$, is open, dense, and of full measure in $\mH_n$ ~\cite{GoWa11, Wa16}. Moreover, the stabilizer of any critical state is a symmetric subgroup of $GL(\mH_n)$, i.e. it is Zariski-closed (Z-closed) (see e.g. \cite{Wa16} for the definition of the Zariski topology) and invariant under the adjoint \cite{GoKr16}. The latter means that, if $g \in \widetilde{G}_\psi$, for $\ket{\psi} \in Crit(\mH_n)$, then $g^\dagger \in \widetilde{G}_\psi$.

Let us now briefly recall how some of us proved in \cite{GoKr16} that there exists an open and full measure set of states in the Hilbert space corresponding to $n$-qubit states with $n\geq 5$, which contains only states with trivial stabilizer in $\widetilde{G}$. In order to do so, we define the following subset of critical states,
\be\label{C0}
\mC\eqdef\left\{\ket{\psi} \in Crit(\mH_n)\;\Big|\;\dim (G|\psi\ra)=\dim(G)\right\}.
\ee
That is, $\mC$ consists of all critical states whose orbits (under $G$) have maximal dimension (i.e. the dimension of $G$). Due to the identity $G|\psi\ra \cong G/G_\psi$ it follows that $|\psi\ra\in\mC$ iff $|\psi\ra$ is critical and $G_\psi$ is a finite group (or equivalently $\dim(G_{\psi})=0$).
Using algebraic geometry and the theory of Lie groups some of us showed in \cite{GoKr16} the following important properties of this subset.

\begin{lemma}{\rm \cite{GoKr16}}\label{properties}
The set $\mC$ defined in Eq.(\ref{C0}) has the following properties:
\begin{enumerate}
\item[(i)] $G_\psi=K_\psi$ for all $|\psi\ra\in\mC$.
\item[(ii)] The set $G\mC \eqdef \left\{g|\psi\ra\;|\;g\in G\;;\;|\psi\ra\in\mC\right\}$ is open with complement of lower dimension in $\mH_n$.
\item[(iii)] $\mC$ is a connected smooth submanifold of $\mH_n$, and $K$ acts differentiably on $\mC$.
\end{enumerate}
\end{lemma}

The principal orbit type theorem (\cite{Bredon}, see also Appendix \ref{app:theorems}) was then central to the proof that the set of states whose stabilizer in $G$ is trivial is open and of full measure. Defining the set
\be
\label{mC0} \mC^0= \left\{|\psi\ra\in\mC\;\big|\;G_\psi=\{\one\}\right\}
\ee
we proved that if $\mC^0$ is not empty, then $\mC^0$ is open, dense, and of full measure in $\mC$ \cite{GoKr16}.
Moreover, in this case, the set
$G\mC^0=\left\{g|\psi\ra\;\big|\;|\psi\ra\in\mC^0, \ g\in G\right\}$ is open, dense, and of full measure in $G\mC$. Clearly, any state $\ket{\phi}$ in $G\mC^0$ has a trivial stabilizer in $G$. Using now that $G\mC$ is open and of full measure in $\mH_n$ (see property~2 in Lemma~\ref{properties}), we also have that $G\mC^0$, which contains only states with $G_\psi=\{\one\}$, is open and of full measure in $\mH_n$. As can be seen from the proofs in \cite{GoKr16} this result holds for arbitrary multipartite quantum systems (as long as it can be shown that $\mC^0$ is not empty). In particular, we have \footnote{This result for the stabilizer in $G$ follows also from Luna's Etal slice theorem (see \cite{slice}, Proposition 5.7). However, since we are interested in the stabilizer in $\widetilde{G}$ the slice theorem does not apply to the case at hand without arguments as given in Sec. \ref{sec:gentriv}.}
\begin{lemma} \label{Lemma2}
If there exists a state $\ket{\psi} \in \mC^0$, then the set \bea \left\{\ket{\phi}\in \mH_n \;\big|\; G_\phi\neq\{\one\}\right\}\eea is of measure zero in $\mH_n$.
\end{lemma}

For $n\geq 5$ we presented in \cite{GoKr16} a $n$--qubit state $\ket{\psi}$, which is contained in $\mC^0$. Hence, for $n\geq 5$ a generic $n$--qubit state has only a trivial stabilizer (in $G$).
In order to define the set $\mA$ containing states with trivial stabilizer in $\tilde{G}$ (not only $G$) which is also open and with complement of lower dimension in $\mH_n$, we used for the qubit case homogeneous SL--invariant polynomials (SLIPs) \cite{GoWa13}. With these SLIPs we were able to identify a full measure subset $\mA\subset G\mC^0$ with the desired property that for any state $\ket{\phi}\in \mA$, $\tilde{G}_\phi =\{\one\}$.

As mentioned before, Lemma \ref{Lemma2} holds for arbitrary qudit-systems. Note further that criteria for when the stabilizer in $G$ of a generic state is trivial were found in \cite{bryan2}. However, in order to obtain the strong implications in entanglement theory (see Sec. \ref{sec:main} and \ref{sec:picture}) it is required to prove that the stabilizer in $\tilde{G}$ (and not only in $G$ ) is trivial. Hence, the last step of the proof, as outlined above, is essential here. However, it is precisely this step, which cannot be easily generalized to arbitrary local dimensions. Hence, we employ new proof methods in the subsequent section to prove directly the existence of a set ${\cal A}$, which contains only states whose stabilizer in $\widetilde{G}$ is trivial, and which is open and of full measure in $\mH_n$.

\subsection{Genericity of states with trivial stabilizer}
\label{sec:gentriv}

Using Lemma \ref{properties}, we prove now one of the main results of this paper. We have already presented Theorem \ref{IntroMain} and its profound implications in entanglement theory in Sec. \ref{sec:main}.  We show in the following that, in order to prove Theorem \ref{IntroMain} for given values of $n$ and $d$, it will eventually be enough to find only \emph{one} critical state $\ket{\psi} \in \mH_{n,d}$ with trivial \emph{unitary} stabilizer, i.e. with $\widetilde{K}_\psi = \{\one\}$.\\ 

Let us first give an outline of the proof of Theorem \ref{IntroMain}.\\
First, we consider the set of critical states, $Crit(\mathcal{H}_{n})$ (see also Eq. (\ref{eq:critstates})). We show that if a critical state has only finitely many local unitary symmetries then there exists no further local (non-unitary) symmetry of this state (see Lemma \ref{LemmaUnitStab}). We then use this together with the results from \cite{GoKr16} and tools from geometric invariant theory to prove the following statement (see Theorem \ref{thmain}). If there exists {\it one} critical state $\ket{\psi} \in \mH_n$ with trivial \emph{unitary} stabilizer, then there exists a set $\mA \subset \mH_n$ of states with trivial stabilizer in $\tilde{G}$ that is open and of full measure in $\mH_n$. Due to this theorem it is sufficient to find one criticial state with trivial stabilizer in $\mathcal{H}_{n,d}$ to proof Theorem \ref{IntroMain} for these values of $n$ and $d$. Finally, we explicitly construct such states and therefore complete the proof of Theorem \ref{IntroMain} (see Secs. \ref{sec:states}, App. \ref{app:trivstab}).\\

Let us now present the details of the proof of Theorem \ref{IntroMain}. 
We first show that the set of critical states with finite stabilizer in $\widetilde{K}$ coincides with the set of critical states with finite stabilizer in $\widetilde{G}$, as stated in the following lemma.

\begin{lemma} \label{LemmaCritFin}
The following subset of critical states, \bea \label{Eq_tildemC} \tilde{\mC}= \left\{\ket{\psi} \in Crit(\mH_n)\;\Big|\;\dim (\widetilde{K}\ket{\psi})=\dim(\widetilde{K})\right\}\eea coincides with the set \bea \left\{\ket{\psi} \in Crit(\mH_n)\;\Big|\;\dim (\widetilde{G}\ket{\psi})=\dim(\widetilde{G})\right\}.\eea
\end{lemma}

\begin{proof}
This lemma is a direct consequence of a much stronger theorem (Theorem 2.12) proven in \cite{Wa16}. This theorem states that if $H$ is a symmetric subgroup of $GL(\mH_n)$ and the so--called maximal compact subgroup of $H$ is $K' = H \cap U(\mH_n)$, then $Lie(H)=Lie(K')+iLie(K')$. That is, $H=K_1 e^{k_2}$, where $K_1\in K'$ and $k_2\in Lie(K')$. In \cite{GoKr16} it was shown that $\widetilde{G}_{\psi}$ is a symmetric subgroup of $GL(\mH_n)$. As shown in \cite{Wa16}, $\widetilde{K}$ is a maximal compact subgroup of $\widetilde{G}$ and so is $\widetilde{K}_\psi$ of $\widetilde{G}_\psi$ \footnote{This follows from Corollary 2.13 in \cite{Wa16}, which states that if $H$ is a symmetric subgroup of $GL(n,\C)$, then $H$ is the Z-closure
of $K' = H\bigcap U(\mH_n)$. In particular, $K'$ is a maximal compact subgroup of $H$.}. Thus, we have that $Lie(\widetilde{G}_\psi)=Lie(\widetilde{K}_\psi)+iLie(\widetilde{K}_\psi)$. Hence, if $\widetilde{K}_\psi$ is finite, then also $\widetilde{G}_\psi$ is finite. Using now that $H\ket{\psi} \cong H/H_\psi$, for $H=\widetilde{G}, \widetilde{K}$, we obtain that $\tilde{C}$ coincides with the set of critical states whose stabilizer is finite in $\widetilde{G}$, which proves the assertion.
\end{proof}

Using the lemma above we are now in the position to prove that if a critical state has a finite stabilizer in $\widetilde{K}$ (or equivalently in $\widetilde{G}$), then all symmetries in $\widetilde{G}$ are unitary. That is, we prove now the following lemma.
\begin{lemma} \label{LemmaUnitStab}
For any state $\ket{\psi}\in \tilde{\mC}$, with $\tilde{\mC}$ given in Eq. (\ref{Eq_tildemC}), it holds that 
\begin{equation*} 
\widetilde{K}_{\psi}=\widetilde{G}_{\psi}.
\end{equation*}
\end{lemma}

\begin{proof}
Due to Lemma \ref{LemmaCritFin} we have that $\tilde{\mC}$ is a subset of $\mC$. Hence, Lemma \ref{properties} (i) implies that for any state $\ket{\psi}\in \tilde{\mC}$, $G_\psi=K_\psi$. To prove now that this equivalence also holds for  $\widetilde{K}_{\psi}$ and $\widetilde{G}_{\psi}$ we consider $\ket{\psi}\in \tilde{\mC}$ and $g \in \widetilde{G}_\psi$. We show that $g$ must be unitary. The Hilbert-Mumford theorem (see e.g. \cite{mumford,Wa16}) implies that for any critical state, $\ket{\psi}$, there exists a homogeneous SLIP $f$ of some degree $m$ such that $f(\ket{\psi})\neq 0$. Now, if $g\in\tilde{G}_{\psi}$ we can write it as $g=zg'$, where $0\neq z\in\mathbb{C}$ and $g'\in G$. Hence, $f(\ket{\psi})=f(g\ket{\psi})=z^mf(g'\ket{\psi})=z^mf(\ket{\psi})$. As $f(|\psi\ra)\neq 0$ this implies that $z^m=1$. Using now the polar decomposition of $g$, i.e. $g=u \sqrt{g^\dagger g}$, with $u\in \widetilde{K}$, the Kempf--Ness theorem (see Appendix \ref{app:theorems}) implies, as $g^\dagger g$ is positive, that $\sqrt{g^\dagger g}\ket{\psi}=\ket{\psi}$. Hence, also $u$ has to be a stabilizer of $\ket{\psi}$. In particular, $u\in \widetilde{K}_\psi$. Moreover, as $z$ is only a phase, we have that $\sqrt{g^\dagger g}=\sqrt{(g')^\dagger (g')}\in G_\psi$. Using now that $G_\psi=K_\psi$ we have that $g=u \sqrt{g^\dagger g}\in \widetilde{K}_\psi$, which proves the statement.
\end{proof}

With Lemma \ref{LemmaUnitStab} it is easy to see that the following sets all coincide:
\bi \item[(i)] \bea \label{mCprime} \mC^\prime=\left\{\ket{\psi} \in \mC\;\Big|\; \widetilde{K}_\psi=\{\one\}\right\}.\eea
\item[(ii)] \bea \left\{\ket{\psi} \in Crit(\mH_n)\;\Big|\; \widetilde{K}_\psi=\{\one\}\right\}.\eea
\item[(iii)] \bea \left\{\ket{\psi} \in Crit(\mH_n)\;\Big|\; \widetilde{G}_\psi=\{\one\}\right\}.\eea
\ei

We now use these results to prove the following theorem, which states that, if $\mC^\prime$ is non-empty then our main theorem, Theorem \ref{IntroMain}, is implied.

\begin{theorem}\label{thmain}
If there exists a state $\ket{\psi} \in \mC^\prime$, i.e. if there exists a critical state $\ket{\psi}$ such that $\widetilde{K}_\psi=\{\one\}$, then there exists an open and full--measure (in ${\cal H}_n$) set of states whose stabilizer in $\widetilde{G}$ is trivial. More precisely, if $\mC^\prime \neq \emptyset$, then the set of states
\bea \label{EqA}
{\cal A}=G \mC^\prime =\left\{g\ket{\psi} \ | \ \ket{\psi} \in \mC^\prime, g \in G \right\},
\eea
which contains only states with trivial stabilizer in $\widetilde{G}$, is open and of full measure in ${\cal H}_n$.\end{theorem}

\begin{proof} First of all, note that $\widetilde{K}$ is a compact Lie group, which acts differentiably on the connected smooth submanifold $\mC$ of ${\cal H}_n$ (see (iii) of Lemma \ref{properties}). Hence, the principal orbit type theorem (see Appendix \ref{app:theorems}) can be applied. This theorem implies that the set
\begin{align}
\mC^\prime=\left\{\ket{\psi} \in \mC\;\Big|\; \widetilde{K}_\psi=\{\one\}\right\}=\left\{\ket{\psi}\in \mC\;\Big|\; \widetilde{G}_\psi=\{\one\}\right\} \label{eq:Cprime}
\end{align}
is, if it is non-empty, open and of full measure in $\mC$. Note that, in the last equality in Eq. (\ref{eq:Cprime}) we used Lemma \ref{LemmaUnitStab}. According to (ii) of Lemma \ref{properties} we have that $G\mC$ is open and of full measure in $\mH_n$. Therefore, for any open and full measure set of $\mC$ the union of the orbits of all states in this set is also open and full measure in $\mH_n$. Using now that for any $\ket{\phi}\in {\cal A}$ there exist $g \in G$ and $\ket{\psi} \in \mC^\prime$ such that $\widetilde{G}_\phi= g \widetilde{G}_\psi g^{-1}$, we have that for any $\ket{\phi}\in {\cal A}$ it holds that $\widetilde{G}_\phi=\{\one\}$, which completes the proof. \end{proof}

In the subsequent section we explicitly present states in $\mC'$ for the Hilbert spaces specified in Theorem \ref{IntroMain}, which completes the proof of this theorem. In Sec. \ref{sec:main} the implications of this result in the context of entanglement theory are discussed. Let us stress here that our results also encompass the results of \cite{GoKr16}, where it was shown that almost all $(n>4)$-qubit states have trivial stabilizer. While no example of a state with trivial stabilizer was given in \cite{GoKr16}, our results allow us to construct states with this property, as we show in the following section.  Let us further remark here that Theorem \ref{thmain} holds for arbitrary multipartite quantum systems. However, in this work we only use it for homogeneous systems, i.e. systems composed of subsystems with equal dimension.

\subsection{Critical states with trivial stabilizer in $\widetilde{G}$}
\label{sec:states}

In this section we present critical states with trivial stabilizer.
First, we introduce a critical state which is defined for $n=5$, $n > 6$, and $d \geq 2$ and give an outline of the proof that its stabilizer is trivial. The proof itself is given in Appendix \ref{app:trivstab}.
It will become evident from the construction of this state that the cases $n=3,4,6$ have to be treated separately. However, also for these cases we construct states with trivial stabilizer in Appendix \ref{app:trivstab}. This completes the proof of Theorem \ref{IntroMain} as it shows that the set $\mC^\prime$ is non-empty for the systems mentioned in this theorem.\\

Let us introduce the following notation before we define the state with the desired properties for $n = 5, n > 6$ and $d\geq 2$. Let $S_n$ denote the symmetric group of $n$ elements. For a permutation $\sigma \in S_n$ we define the operator $P_{\sigma}$ via
$P_{\sigma}\ket{i_1}\otimes \ldots \otimes \ket{i_n} = \ket{i_{\sigma^{-1}(1)}}\otimes \ldots \otimes \ket{i_{\sigma^{-1}(n)}}$, for all $(i_1,\ldots,i_n)\in\{0,\ldots,d-1\}^n$.
We call a state $\ket{\psi}$ symmetric if $P_\sigma \ket{\psi}=\ket{\psi}$ for all $\sigma \in S_n$. Furthermore, we define for an arbitrary state $\ket{\phi} \in \mH_{n,d}$ the set of all distinct permutations of $\ket{\phi}$ as $\pi(\ket{\phi}) = \{P_{\sigma}\ket{\phi} \ | \ \sigma \in S_n\}$ and the symmetrization of $\ket{\phi}$ as
$\ket{\pi(\ket{\phi})} = \sum_{\ket{\chi} \in \pi(\ket{\psi})} \ket{\chi}$.
Using this notation, we define for $0\leq k \leq n$ and $j\in \{1,\ldots,d-1\}$ the (unnormalized) state
\begin{align*}
 \ket{D_{k,n}(j)} = \ket{\pi(\ket{j}^{\otimes k}\ket{j-1}^{\otimes n-k})}.
\end{align*}
We are now ready to introduce the critical $n$-qudit state ($n=5, n>6$) for which we show that it has trivial stabilizer, namely
\begin{align}
 \ket{\Psi_{n,d}} &= \sum_{j=0}^{d-1} c_j \ket{j}^{\otimes n} + \sum_{j=1}^{d-1}\ket{D_{k,n}(j)}, \label{eq:state}
\end{align}
where $c_0 = \sqrt{{{n-1}\choose{k-1}} + 1}$, $c_i = 1$ for $0<i<d-1$ and $c_{d-1} = \sqrt{{{n-1}\choose{k}} + 1}$, with $k$ the smallest natural number such that $3 \leq k \leq n-2$, $n \neq 2k$ and $\text{gcd}(n,k) = 1$. Here, $\text{gcd}(n,k)$ denotes the greatest common divisor of $n$ and $k$. The existence of
$k$ is obvious for $n=5$ and is proved for $n > 6$ in Appendix \ref{app:trivstab}. The condition $\text{gcd}(n,k) = 1$ is crucial to ensure that the state $\ket{\Psi_{n,d}}$ has only trivial symmetries as we shall see later. Recall that $\ket{\Psi_{n,d}}$ is critical if all of its single-subsystem reduced states are proportional to the identity. A straightforward calculation shows that $\ket{\Psi_{n,d}}$ indeed fulfills this property for $n=5$, $n > 6$. Note further that $\ket{\Psi_{n,d}}$ is not defined for $n=2,3,4,6$ since there is no $k$ with the properties described after Eq. (\ref{eq:state}).\\

Note that for $(n>4)$-qubits the existence of states with trivial stabilizer was shown in \cite{GoKr16}. However, no examples of states with this property were given. The mathematical methods developed in this article allow us to explicitly construct such states, namely the states $\{\ket{\Psi_{n,2}}\}_{n=5,n>6}$ \footnote{A 6-qubit state with trivial stabilizer is given in Appendix \ref{sec:n6}.}. This shows that our work also includes, and in fact extends, the results for qubits obtained in \cite{GoKr16}. To explicitly give an example of a qubit state with trivial stabilizer, consider the 5-qubit state,
\begin{align*}
 \ket{\Psi_{5,2}} =& \sqrt{7}\ket{00000} + \ket{00111} + \ket{01011} + \ket{01101}\\
 &+\ket{01110} + \ket{10011} + \ket{10101} + \ket{10110}\\
 &+ \ket{11001} + \ket{11010} + \ket{11100} + \sqrt{5} \ket{11111}.
\end{align*}

The following lemma shows that $\ket{\Psi_{n,d}}$ has trivial stabilizer for $n=5,n> 6$ and $d \geq 2$. It is a combination of Lemma \ref{lem:trivstab7} and Lemma \ref{lem:trivstabn5} in Appendix \ref{sec:proofsn}, where we prove this statement for $n>6$ and $n = 5$, respectively.
\begin{lemma}
\label{lem:trivstab}
 For $n = 5, n > 6$ and $d \geq 2$ the stabilizer of $\ket{\Psi_{n,d}}$ is trivial, i.e. $\widetilde{G}_{\Psi_{n,d}} = \{\one\}$.
\end{lemma}
In the following we give an outline of the proof of this lemma, which is divided into four main steps.\\

First, we note that it is sufficient to show that $\widetilde{K}_{\Psi_{n,d}} = \{\one\}$ as Lemma \ref{LemmaUnitStab} then implies that also $\widetilde{G}_{\Psi_{n,d}} = \{\one\}$ holds.
In the second step, we show that any $v \in \widetilde{K}_{\Psi_{n,d}}$ is of the form $v = u^{\otimes n}$ for some $u \in U(d)$. The proof of this statement is presented in Appendix \ref{sec:proofsn}. It thus remains to show that the only $u\in U(d)$ that fulfills the equation
\begin{align}
 u^{\otimes n} \ket{\Psi_{n,d}} = \ket{\Psi_{n,d}} \label{eq:EVnmain}
\end{align}
also fulfills $u^{\otimes n} = \one$.
In the third step, we show that Eq. (\ref{eq:EVnmain}) can only be fulfilled if $u$ is diagonal, i.e. if $u = \sum_i u_i \ket{i}\bra{i}$. We show this in Appendix \ref{sec:proofsn} by considering the 2-subsystem reduction of Eq. (\ref{eq:EVnmain}).
In the fourth step, we reinsert $u = \sum_i u_i \ket{i}\bra{i}$ into Eq. (\ref{eq:EVnmain}) and see that it is equivalent to
\begin{align}
& u_i^n = 1 \ \text{for} \ i \in \{0,\ldots,d-1\}, \label{eq:phases1main}\\
& u_i^k u_{i-1}^{n-k} = 1 \ \text{for} \ i \in \{1,\ldots,d-1\} \label{eq:phases2main}.
\end{align}
Now, recall that $\text{gcd}(n,k) = 1$. This can be used to show that the only solution of Eqs. (\ref{eq:phases1main}-\ref{eq:phases2main}) is $u = \omega_n^m \one$, where $\omega_n = \exp(2\pi i/n)$ and $m\in\mathbb{N}$.
Hence, $u^{\otimes n} = \one$ holds. This completes the proof of Lemma \ref{lem:trivstab}.

\section{Multipartite pure state transformations}
\label{sec:picture}

Combining our result with previous works, the following picture of multipartite pure state entanglement transformations emerges (see Fig. \ref{fig:1}). For bipartite pure states (region A in Fig. \ref{fig:1}) all deterministic and probabilisitc LOCC transformations are characterized \cite{nielsen, Vidal99} and SEP $=$ LOCC holds \cite{gheorg}. Entangled bipartite pure states can always be transformed via nontrivial deterministic LOCC, regardless of their local dimensions. Moreover, they can always be obtained from the maximally entangled state. Hence, this (up to LUs) single state constitutes
the maximally entangled set of bipartite states.\\
\indent Three qubits (region B in Fig. \ref{fig:1}) are the only multipartite system for which all deterministic LOCC transformations between pure states are characterized \cite{turgut}. They are moreover the only tripartite system for which it is known that all fully entangled pure states can be transformed to other fully entangled states via nontrivial deterministic LOCC. Moreover, SEP $=$ LOCC for deterministic transformations within the GHZ class, i.e. for deterministic transformations between generic states \cite{chen}. Furthermore, the MES, i.e. the minimal set from which all other states can be deterministically obtained via LOCC, is of measure zero, albeit uncountably infinite \cite{dVSp13}. This situation changes drastically when the local dimension is increased by only one.\\
\indent Generic 3-qutrit states (region C in Fig. \ref{fig:1}) are isolated, despite the fact that their stabilizer is nontrivial \cite{HeSp15}. The MES is of full-measure. Moreover, SEP $\neq$ LOCC for deterministic transformations of 3-qutrit pure states \cite{HeSp15}.\\
\indent Regarding three-partite states we show that, already for  $4-,5-,$ or $6$-level systems (in region D in Fig. \ref{fig:1}), almost all pure states have trivial stabilizer and are therefore isolated (see Theorem \ref{IntroMain}). We further derive the optimal probabilistic protocol for transformations between generic states and find that SEP $=$ LOCC for these conversions. An open question is whether these results extend to tripartite systems of any local dimension $d>3$ (region E in Fig. \ref{fig:1}) or not.\\
\indent Four-qubit pure states (region F in Fig. \ref{fig:1}) generically have finite, nontrivial stabilizer and their MES is of full measure \cite{dVSp13,SpdV16}. Furthermore, SEP $=$ LOCC for transformations among generic pure states, which were characterized in \cite{Sa15}. However, almost all states are isolated \cite{dVSp13}.\\
\indent Finally, our work shows that almost all qudit states of $n$ $d$-level systems with $n=3$ and $d=4,5,6$ or $n>3$ and $d \geq 3$ (region D in Fig. \ref{fig:1}) have trivial stabilizer (see Theorem \ref{IntroMain}) and are therefore isolated. That is, almost all qudit states are in the MES. We further determine the optimal protocol for probabilistic transformations among these states and find that SEP $=$ LOCC holds in these cases. This shows in particular that the results of \cite{GoKr16}, which is devoted to $(n>4)$-qubit systems (region G in Fig. \ref{fig:1}), can be generalized to arbitrary local dimension.

\begin{figure}
 \includegraphics[width=0.4\textwidth]{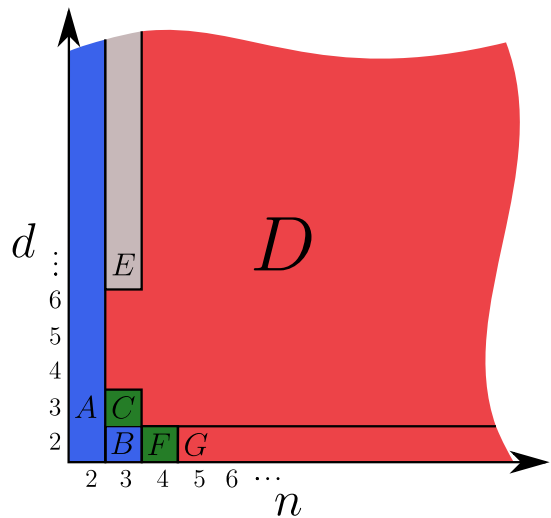}
 \caption{Summary of results on the symmetries of $n$-partite systems with local dimension $d$. The picture is divided into different regions (A to G) that were treated separately in the literature.
 The colors give information on the stabilizer of states in the corresponding system: blue (all states have non-compact stabilizer; regions A \cite{nielsen}, B \cite{turgut, dVSp13}), green (generic states have finite, nontrivial stabilizer; regions C \cite{HeSp15}, F \cite{dVSp13}), red (generic states have trivial stabilizer;  D (see Th. \ref{IntroMain}), G \cite{GoKr16}), grey (generic states have finite stabilizer \cite{GoWa11,Wa16}; unknown if it is trivial; region E). The implications of these and other results in entanglement theory are summarized in the main text.}
  \label{fig:1}
\end{figure}

\section{Conclusion}
\label{sec:conclusion}
In this work, we used methods from geometric invariant theory and the theory of Lie groups to prove that almost all pure $(n>3)$-qudit states and almost all three $d$-level states, for $d=4,5,6$, have trivial stabilizer. Combined with the characterization of local transformations of states with trivial stabilizer provided in \cite{GoKr16}, this has profound implications in entanglement theory. It allows us to characterize all transformations via LOCC and via SEP among almost all $(n>3)$-qudit pure states. We find that these transformations are extremely restricted. In fact, almost all $(n>3)$-qudit pure states are isolated. Due to the results presented here, the simple expression for the optimal success probability for probabilistic local transformations presented in \cite{GoKr16} is shown to hold among generic states. The optimal SEP protocol is a so called one-successful-branch protocol (OSBP), i.e. a simple protocol for which only one branch leads to the final state, which can also be implemented via LOCC. Furthermore, we discussed implications of our result for the construction of entanglement measures, the characterization of LU-equivalence classes and for the determination of probabilistic multi-copy LOCC transformations of multiqudit pure states. All of these results also hold for three $d$-level systems, where $d=4,5,6$.

This work shows that, in the context of local state transformations, only a zero-measure subset of the exponentially large space of $(n>3)$-qudit states is physically significant. That is, the most powerful states are very rare. This is consistent with investigations in other fields of physics, e.g. condensed matter physics, where it has been shown that under certain conditions only a zero-measure subset of all quantum states is physically relevant \cite{ScPe10} . These results therefore suggest that the physically relevant zero-measure subset of states, such as matrix-product states \cite{ScPe10}, projected-entangled pair states \cite{peps} (with low bond dimension), or stabilizer states \cite{GoThesis}, should be investigated more deeply. As transformations between fully entangled states of homogeneous systems are almost never possible, it would moreover be interesting to study transformations of generic states of heterogeneous systems. The methods developed in Sec. \ref{sec:gentriv} can be applied to arbitrary multipartite systems. However, interestingly, for certain heterogeneous systems, one can show that generic states always have nontrivial local symmetries \cite{HeGa17}. Our results further suggest that more general local transformations should be considered. This includes the multi-copy case and transformations between states of different local dimensions or number of subsystems e.g. transformations from $n$-qubit states to $(n-k)$-qubit state, where $1 \leq k \leq n-1$. Finally, the fact that almost all qudit states have trivial stabilizer and the mathematical tools that we developed to prove this could also be relevant in other fields of physics, such as condensed matter physics.\\

\noindent\textit{Acknowledgments:---}   G.G. research is supported by the Natural Sciences and Engineering Research Council of Canada (NSERC). The research of D.S. and B.K. was funded by the Austrian
Science Fund (FWF) through Grants No. Y535-N16 and No. DK-ALM:W1259-N27. D.S. would like to thank B.K., G.G. and the DK-ALM for giving him the opportunity to visit G.G.'s research group for two months while working on this project. D.S. would also like to thank the research group of G.G. for their hospitality during his visit.

\appendix

\section{Kempf-Ness theorem and Principal orbit type theorem}
\label{app:theorems}

In this Appendix we review two theorems that are central to our work, the Kempf-Ness theorem and the principal orbit type theorem, and discuss where they are used in the main text. We also briefly review some implications of the Kempf-Ness theorem in the context of entanglement theory. For further details the interested reader is referred to \cite{GoKr16} and \cite{Wa16}.\\

Let us first review the Kempf-Ness theorem.
In the main text we introduced the set of critical states in Eq. (\ref{eq:critstates}) as
\begin{align*}
 Crit(\mathcal{H}_{n})\equiv\{\ket{\phi}
\in\mathcal{H}_{n}|\left\langle \phi|X|\phi\right\rangle =0, \forall X \in
Lie(G)\}.
\end{align*}
Note that a state $\ket{\psi} \in\mathcal{H}_n$ is critical iff all of its local density matrices are proportional to the identity \cite{GoWa11}. That is, a state is critical if every subsystem is maximally entangled with the remaining subsystems.
As mentioned in the main text, many well-known quantum states are critical, and the union of the $G$-orbits of all critical states is dense and of full-measure in $\mH_n$. Critical states have many other interesting properties  mentioned below. Some of these can be derived from the Kempf-Ness theorem, which provides a characterization of critical states.
\begin{theorem}{\rm \cite{KN}}\label{KN}{\rm ~ \bf The Kempf-Ness theorem}
\begin{enumerate}
\item $\ket{\phi}\in Crit(\mathcal{H}_{n})$ iff $\left\Vert g\ket{\phi}\right\Vert \geq\left\Vert \ket{\phi}\right\Vert $ for all
$g\in G$.
\item If $\ket{\phi}\in Crit(\mathcal{H}_{n})$ and $g\in G$ then $\left\Vert g\ket{\phi}\right\Vert
\geq\left\Vert \ket{\phi}\right\Vert $ with equality iff $g\ket{\phi}\in K\ket{\phi}$.
Moreover, if $g$ is positive definite then the equality condition holds iff
$g\ket{\phi}=\ket{\phi}$.
\item If $\ket{\phi}\in\mathcal{H}_{n}$ then $G\ket{\phi}$ is closed in $\mathcal{H}_{n}$ iff $G\phi\cap Crit(\mathcal{H}_{n})\neq\emptyset$.
\end{enumerate}
\end{theorem}
The second part of the theorem implies that each SLOCC orbit contains (up to local unitaries) at most one critical state. Thus, critical states are natural representatives of SLOCC orbits. They are the unique states in their SLOCC orbits for which each qubit is maximally entangled to the other qubits \cite{GoWa11}.
The Kempf-Ness theorem was also important in the proof of \cite{GoKr16} to show that $g \in \tilde{G}_\psi$ iff $g^\dagger \in \tilde{G}_\psi$ for a critical state $\ket{\psi}$. Together with the fact that $\tilde{G}_\psi$ is Z-closed (which follows from the definition) this shows that $\tilde{G}_\psi$ is a symmetric subgroup of $GL(\mH_n)$ (see e.g. \cite{Wa16} for the definition of the Zariski topology). This property is central to the proof of Lemma \ref{LemmaUnitStab} in this work.\\

In order to state the principal orbit type theorem we first introduce some definitions and notation. We further discuss how a subgroup $H \subset GL(\mH_n)$ induces a preorder on the the set of all $H$-orbits of states in $\mH_n$.
The principal orbit type theorem then provides conditions under which this preorder gives rise to a maximal element.

Let $\ket{\psi},\ket{\phi} \in \mH_n$ be two states. Then $H_\psi$ and $H_\phi$ are said to have the same \emph{type} if there exists a
$h\in H$ such that $H_\phi=h H_\psi h^{-1}$, i.e. if they are conjugate in $H$.
Clearly, the stabilizer of $\ket{\psi}$ and $h\ket{\psi}$ are conjugate for any $h \in H$, namely
\begin{equation*}
hH_\psi h^{-1}=H_{h\psi}.
\end{equation*}
Hence, $H_\psi$ and $H_\phi$ are of the same type iff there exists $h\in H$ such that $H_\phi=H_{h\psi}$.  However, the fact that $H_\psi$ and $H_\phi$ have the same type does \emph{not} imply that there is a $h \in H$ such that $\ket{\psi} = h \ket{\phi}$, i.e. it does not imply that they are in the same $H$-orbit. For example, the $\widetilde{G}$-stabilizer of generic 4-qubit states has the same type as $\{\sigma_i^{\otimes 4}\}_{i=0}^3$ \cite{dVSp13}, despite the fact that two 4-qubit states are generically SLOCC inequivalent and thus not in the same $\widetilde{G}$-orbit \cite{verstr}.

We further say that $H/H_{\psi}$ has a lower type than $H/H_{\phi}$, denoted as
\begin{equation*}
H/H_\psi \prec_{type} H/H_\phi,
\end{equation*}
if $H_\phi$ is conjugate in $H$ to a subgroup of $H_\psi$. It is easy to see that $\prec_{type}$ induces a preorder on the set of all $H$-stabilizers.

That this preorder also induces a preorder on set of all $H$-orbits can be seen as follows.
Note that $H\ket{\psi}$ is isomorphic to the left coset of $H_\psi$ in $H$ for all $\ket{\psi}$, namely
\begin{equation*}
H\ket{\psi} \cong H/H_\psi.
\end{equation*}
We can therefore say that $H|\psi\ra$ is of lower type than $H|\phi\ra$, denoted as
\begin{equation*}
H\ket{\psi} \prec_{type} H\ket{\phi},
\end{equation*}
if $H/H_\psi \prec_{type} H/H_\phi$ holds.

The following theorem, called principal orbit type theorem (POT theorem), shows that under certain very general conditions this preorder possesses a maximal element. This key theorem can be found in
\cite{Bredon}, as a combination of Theorem 3.1 and Theorem 3.8.
\begin{theorem}{\rm \cite{Bredon}}\label{thm:principal} {\rm \bf The principal orbit type theorem}
Let $C$ be a compact Lie group acting differentiably on a connected smooth
manifold $\mM$ (in this paper we assume $\mM\subset\mH_n$).
Then, there exists a principal orbit type; that is, there exists a state $\ket{\phi} \in \mM$ such that $C/C_\psi\prec_{type} C/C_{\phi}$ for all $|\psi\ra\in \mM$.
Furthermore, the set of $\ket{\psi} \in \mM$ such that $C_{\psi}$ is conjugate to $C_\phi$ is open
and dense in $\mM$ with complement of lower dimension and hence of measure 0.
\end{theorem}
The following example illustrates how powerful the POT theorem is. Suppose $\ket{\psi} \in \mH_n$ is a (not necessarily critical) state with trivial unitary stabilizer, $\widetilde{K}_\psi = \{\one\}$.
Then the POT theorem applied to $C = \tilde{K}$ and $\mM = \mH_{n}$ directly implies that the set of states with trivial unitary stabilizer
is of full measure in $\mH_n$.

However, in this work we show that the stabilizer in $\widetilde{G}$ is generically trivial. As $\widetilde{G}$ is a noncompact Lie group, the POT theorem cannot be applied directly. It is nevertheless central to the proof of Theorem \ref{thmain}, where we applied it to the compact Lie group $C = \tilde{K}$ that acts differentiably on the connected smooth manifold $\mM = \mC$ (see Lemma \ref{properties}, (iii)).\\

\section{Criticial states with trivial stabilizer}
\label{app:trivstab}

In this appendix we provide examples of critical states with trivial stabilizer for $n = 3$ and $d = 4,5,6$, for $n=4$ and $d > 2$ and for $n \geq 5$ and $d \geq 2$. That is, we give examples of criticial states with trivial stabilizer for all Hilbert spaces described in Theorem \ref{IntroMain}, and for $(n>4)$-qubit systems. Combined with Theorem \ref{thmain} this completes the proof of Theorem \ref{IntroMain} that almost all pure states in these Hilbert spaces have trivial stabilizer. For $n>3$ we present these states in Sec. \ref{ngeq3}. The states with $n = 3$ have to be constructed differently and are presented in Sec. \ref{sec:tripartite}.

\subsection{Critical ($n>3$)-qudit states with trivial stabilizer}
\label{ngeq3}

In this section we present critical states with trivial stabilizer for $n = 4, d>2$ and $n>4, d\geq 2$. As we will see, it is easy to show that these states are indeed critical, i.e. that their single-subsystem reduced states are proportional to the completely mixed state. In contrast to that, the proof that their stabilizer is trivial is more involved and the details of the proof depend on $n$ and on $d$.
However, since we consider only permutationally symmetric states the main steps of this proof are the same for all $n>3$. For the sake of readability we outline these four main steps before we present the details in the subsequent subsections. The main ingredients to show that a permutationally symmetric, critical state considered here, say $\ket{\psi_{n,d}} \in \mH_{n,d}$, has trivial stabilizer, are the following.\\
\begin{itemize}
 \item[(1)] Since $\ket{\psi_{n,d}}$ is critical it is sufficient to show that $\widetilde{K}_{\psi_{n,d}} = \{\one\}$ holds, i.e. that $\ket{\psi_{n,d}}$ has a trivial unitary stabilizer, as Lemma \ref{LemmaUnitStab} states that then also $\widetilde{G}_{\psi_{n,d}} = \{\one\}$ holds.
 \item[(2)] We show that a unitary $B$ fulfills $B \otimes B^{-1} \otimes \one^{\otimes n-2} \ket{\psi_{n,d}} = \ket{\psi_{n,d}}$ iff $B = c \one$ for some phase $c$. It was shown in \cite{migdal} that then any $v \in \widetilde{K}_{\psi_{n,d}}$ can be express as $v = u^{\otimes n}$ for some $u \in U(d)$ (see also Lemma \ref{lem:symSym5} below for details).
\item[(3)] It remains to show that for any unitary $u\in U(d)$ the equation
\begin{align}
 u^{\otimes n} \ket{\psi_{n,d}} = \ket{\psi_{n,d}} \label{eq:EVsummary}
\end{align}
implies that $u^{\otimes n} = \one$. The corresponding equation for the reduced state of the first two subsystems, $\rho_{n,d}^{(1,2)}= \tr_{3,\ldots,n}(\ket{\Psi_{n,d}}\bra{\psi_{n,d}})$ reads,
 \begin{align*}
  (u\otimes u) \rho_{n,d}^{(1,2)} (u^{\dagger} \otimes u^{\dagger}) = \rho_{n,d}^{(1,2)}.
 \end{align*}
 This equation can be used to show that $u$ has to be diagonal. However, the details of this proof depend on $n$ and $d$.
\item[(4)] In the last step we show that the only diagonal unitary $u$ that fulfills Eq. (\ref{eq:EVsummary}) also fulfills $u^{\otimes n} = \one$. This shows that $\widetilde{K}_{\psi_{n,d}} = \{\one\}$ 
and completes the proof.
\end{itemize}

The remainder of this section is devoted to the details of this proof.
In Sec. \ref{sec:proofsn} we consider the case $n=5$ and $n>6$ and $d\geq 2$. In Sec. \ref{sec:n4} we consider the case $n=4, d>2$ and in Sec. \ref{sec:n6} the case $n=6$ and $d\geq 2$.

\subsubsection{A critical $n$-qudit state, $n=5, n > 6$, with local dimension $d \geq 2$ and trivial stabilizer}
\label{sec:proofsn}

In this section we show that the critical state $\ket{\Psi_{n,d}}$ introduced Eq. (\ref{eq:state}) of Sec. \ref{sec:states} is well-defined and has trivial stabilizer for $n=5, n > 6$ and $d\geq2$. That is, we prove Lemma \ref{lem:trivstab} of the main text.\\

Let us first recall the following definitions made in Sec. \ref{sec:states} of the main text. For $\ket{\phi} \in \mH_{n,d}$ we define the set of all distinct permutations of $\ket{\phi}$ as $\pi(\ket{\phi}) = \{P_{\sigma}\ket{\phi} \ | \ \sigma \in S_n\}$ and the symmetrization of $\ket{\phi}$ as
$\ket{\pi(\ket{\phi})} = \sum_{\ket{\chi} \in \pi(\ket{\phi})} \ket{\chi}$.
Using this notation, we define the (unnormalized) state
\begin{align*}
 \ket{D_{k,n}(j)} = \ket{\pi(\ket{j}^{\otimes k}\ket{j-1}^{\otimes n-k})},
\end{align*}
for $0\leq k \leq n$ and $j\in \{1,\ldots,d-1\}$. These states fulfill
\begin{align}
 \langle D_{k,n}(j)|D_{k',n}(j')\rangle = {{n}\choose{k}} \delta_{k,k'} \delta_{j,j'}.
\end{align}
For $l \in \{1,\ldots,n-1\}$ we can express $\ket{D_{k,n}(j)}$ in the bipartite splitting of any $l$ subsystems and the remaining $n-l$ subsystems as
\begin{align}
\label{eq:decDicke}
 \ket{D_{k,n}(j)} = \sum_{q=0}^{\min\{l,k\}} \ket{D_{q,l}(j)}\ket{D_{k-q,n-l}(j)}.
\end{align}

In Sec. \ref{sec:states} we then defined for $n = 5, n > 6$ and $d \geq 2$ the state
\begin{align}
 \ket{\Psi_{n,d}} &= \sum_{j=0}^{d-1} c_j \ket{j}^{\otimes n} + \sum_{j=1}^{d-1}\ket{D_{k,n}(j)}, \label{eq:stateA}
\end{align}
where $c_0 = \sqrt{{{n-1}\choose{k-1}} + 1}$, $c_i = 1$ for $0<i<d-1$ and $c_{d-1} = \sqrt{{{n-1}\choose{k}} + 1}$, with $k$ the smallest natural number such that $3 \leq k \leq n-2$, $n \neq 2k$ and $\text{gcd}(n,k) = 1$.\\

Let us first show that $k$ as described above always exists for $n = 5, n > 6$ and that $\ket{\Psi_{n,d}}$ is therefore well-defined. For $n \in \mathbb{N}$ the Euler totient function $\phi(n)$ is defined as the number of all natural numbers $j$ that are smaller than $n$ and fulfill $\text{gcd}(n,j)=1$, i.e.
\begin{align}
 \phi(n) = |\{j \in \mathbb{N} | \ j < n, \ \text{gcd}(n,j)=1\}|. \label{eq:Euler}
\end{align}
 It is straightforward to see that $k$ as defined below Eq. (\ref{eq:stateA}) always exists if $\phi(n) \geq 5$. We now prove that $k$ exists if $n=4$ and $n > 6$.  If $n$ is not divisible by $3$
then $k=3$. If $n$ is divisible by $3$ and $5$ then Euler's formula for $\phi(n)$ (cf. \cite{nicolas})
implies that $\phi(n)\geq 5$. Finally, if $n$ is divisible by $3$ but not $5$ and $n \geq 9$ then $k=5$.\\

Let us now show some properties of $\ket{\Psi_{n,d}}$ that will be useful in the proof that it has trivial stabilizer.
Note first that, due to Eq. (\ref{eq:decDicke}), we can express $\ket{\Psi_{n,d}}$ in the bipartition of any $l$ subsystems with the rest as
\begin{align}
 \ket{\Psi_{n,d}} &= \sum_{j=0}^{d-1} c_j \ket{j}^{\otimes l} \ket{j}^{\otimes n-l} \nonumber \\
 &+ \sum_{j=1}^{d-1} \sum_{q=0}^{\min\{l,k\}} \ket{D_{q,l}(j)}\ket{D_{k-q,n-l}(j)}. \label{eq:decPsi}
\end{align}

Note further that the following useful lemma on symmetric states has been shown in \cite{migdal}.
\begin{lemma}{\rm \cite{migdal}}\label{lem:B}
Let $\ket{\psi}$ be symmetric. Suppose $|\psi\ra$ has the property that
\begin{align}
 B\otimes B^{-1}\otimes \one^{\otimes n-2} \ket{\psi}=\ket{\psi} \ \text{iff} \ B = b\one, \label{eq:B}
\end{align}
for some $b \in \C\backslash\{0\}$.
If $g \in \tilde{G}_{\psi}$ then $g = h^{\otimes n}$ for some $h \in GL(d,\C)$.
\end{lemma}
This result can be easily understood as follows. Let $\ket{\psi}$ be symmetric and let $P_{(1,2)}$ denote the operator that permutes subsystems 1 and 2. Note that if $g = g_1 \otimes \ldots \otimes g_n \in \widetilde{G}_{\psi}$ then $g^{-1}\ket{\psi} = \ket{\psi}$ and $P_{(1,2)}gP_{(1,2)}(P_{(1,2)}\ket{\psi}) = P_{(1,2)}\ket{\psi}$ hold. Using that $P_{(1,2)}\ket{\psi} = \ket{\psi}$ this implies that $g^{-1}P_{(1,2)}gP_{(1,2)}\ket{\psi} = \ket{\psi}$, i.e.
\begin{align}
 g_1^{-1}g_2 \otimes g_2^{-1}g_1 \otimes \one^{\otimes n-2}\ket{\psi} = \ket{\psi}. \label{eq:B2}
\end{align}
Now, if Eq. (\ref{eq:B}) holds, this implies that there is a $c_{1,2} \in \C\backslash\{0\}$ such that $g_1 = c_{1,2} g_2$. As $\ket{\psi}$ is symmetric the same argument can also be used to show that there is a  $c_{i,j} \in \C\backslash\{0\}$ such that  $g_i = c_{i,j} g_j$ for all $i \neq j$ and thus $g = h^{\otimes n}$ for some $h \in GL(d,\C)$. Note that if $g \in \widetilde{K}$ then $g_1^{-1}g_2$ in Eq. (\ref{eq:B2}) is unitary. Hence, in order to show that any unitary symmetry $v \in \widetilde{K}_{\psi}$ is of the form $v = u^{\otimes n}$ for some $u \in U(d)$ it is thus sufficient to show that Eq. (\ref{eq:B}) holds for any unitary $B$. We use this to prove the following lemma.

\begin{lemma}
 \label{lem:symSym5}
 If $v \in \widetilde{K}$ fulfills $v\ket{\Psi_{n,d}} = \ket{\Psi_{n,d}}$ then $v = u^{\otimes n}$ for some $u \in U(d)$.
\end{lemma}
\proof{
$\ket{\Psi_{n,d}}$ is permutationally symmetric. Due to Lemma \ref{lem:B} it is thus sufficient to show that the only solution to
\begin{align}
 B\otimes B^{-1} \otimes \one^{\otimes n-2} \ket{\Psi_{n,d}} = \ket{\Psi_{n,d}}, \label{eq:Bmatrix2}
\end{align}
where $B \in U(d)$, is $B = b\one$ for some complex number $b \neq 0$.\\

In order to show this, we first apply for $i \in \{0,\ldots,d-1\}$ the operator $\one^{\otimes 2}\otimes \bra{i}^{\otimes n-2}$ to both sides of Eq. (\ref{eq:Bmatrix2}). Decomposition (\ref{eq:decPsi}) for $l = 2$ is very useful in calculating the resulting equation.
For $k < n-2$ we obtain
\begin{align*}
  (B \otimes B^{-1})\ket{i}^{\otimes 2} = \ket{i}^{\otimes 2} \ \text{for} \ i \in \{0,\ldots,d-1\}.
\end{align*}
 In the case of $k=n-2$ we get
\begin{align*}
 &(B \otimes B^{-1})\ket{0}^{\otimes 2} = \ket{0}^{\otimes 2},\\
 &(B \otimes B^{-1}) (c_i \ket{i}^{\otimes 2} + \ket{i-1}^{\otimes 2})\\
 &= c_i \ket{i}^{\otimes 2} + \ket{i-1}^{\otimes 2} \ \text{for} \ i \in \{1,\ldots,d-1\}.
\end{align*}
It is then straightforward to see that these equations can only be fulfilled if $B$ is diagonal, i.e. $B = \sum_{i=0}^{d-1} b_i \ket{i}\bra{i}$.\\

Analogously we can apply $\one^{\otimes 2} \otimes \bra{D_{k-1,n-2}(i)}$ to both sides of Eq. (\ref{eq:Bmatrix2}) for $i \in \{1,\ldots,d-1\}$ and see that
\begin{align*}
 (B\otimes B^{-1})\ket{D_{1,2}(i)} = \ket{D_{1,2}(i)} \ \text{for}  \ i \in \{1,\ldots,d-1\}.
\end{align*}
Using that $\ket{D_{1,2}(i)} = \ket{i}\ket{i-1}+\ket{i-1}\ket{i}$ it is easy to see that these equations imply $b_i = b_j = b$ for all $i,j \in \{0,\ldots,d-1\}$ and hence $B = b\one$.\qed}\\

With these results we are in the position to proof  Lemma \ref{lem:trivstab} in the main text, namely that $\ket{\Psi_{n,d}}$ has trivial stabilizer for $n>4, n\neq 6, d > 2$. We first prove this result for $n>6$ (see Lemma \ref{lem:trivstab7}). After that, we provide the proof for $n=5$ (see Lemma \ref{lem:trivstabn5}) which is slightly different.\\

We first consider the case of $n>6$ qubits and prove the following lemma.
\begin{lemma}
\label{lem:trivstab7}
 For $n > 6$ and $d\geq 2$ the stabilizer of $\ket{\Psi_{n,d}}$ is trivial, i.e. $\widetilde{G}_{\Psi_{n,d}} = \{\one\}$.
\end{lemma}

\proof{
Due to Lemma \ref{LemmaUnitStab} and Lemma \ref{lem:symSym5} it remains to show that for any unitary $u\in U(d)$ the equation
\begin{align}
 u^{\otimes n} \ket{\Psi_{n,d}} = \ket{\Psi_{n,d}} \label{eq:EVn}
\end{align}
implies that $u^{\otimes n} = \one$. We first show that Eq. (\ref{eq:EVn}) implies that $u$ is diagonal and then that $u^{\otimes n} = \one$.

Considering the reduced state of the first two subsystems, $\rho_{n,d}^{(1,2)}= \tr_{3,\ldots,n}(\ket{\Psi_{n,d}}\bra{\Psi_{n,d}})$, Eq. (\ref{eq:EV5}) implies that
 \begin{align}
  (u\otimes u) \rho_{n,d}^{(1,2)} (u^{\dagger} \otimes u^{\dagger}) = \rho_{n,d}^{(1,2)}. \label{eq:reduced}
 \end{align}
Note that for $n>6$ we have that $k < n-2$, where $k$ is defined below Eq. (\ref{eq:stateA}) \footnote{This can again be shown using Euler's totient function, $\phi(n)$ (see Eq. (\ref{eq:Euler})).}. This fact simplifies the computation of $\rho_{n,d}^{(1,2)}$ (in contrast to the case $n=5$, where $k=3$; see Lemma \ref{lem:trivstabn5}).
It is then easy to see that
\begin{align}
\rho_{n,d}^{(1,2)} &= \alpha (\ket{0}\bra{0}^{\otimes 2} + \ket{d-1}\bra{d-1}^{\otimes 2})+\beta \sum_{j=1}^{d-2} \ket{j}\bra{j}^{\otimes 2} \nonumber \\
&+ \gamma \sum_{j=1}^{d-1} \ket{D_{1,2}(j)}\bra{D_{1,2}(j)} \label{Eq_Specdec}
\end{align}
with
\begin{align*}
 &\alpha = {{n-1}\choose{k-1}}+{{n-2}\choose{k}}+1,\\
 &\beta = {{n-2}\choose{k}}+{{n-2}\choose{k-2}}+1,\\
 &\gamma = {{n-2}\choose{k-1}}.
\end{align*}
Using Eq. (\ref{eq:reduced}) we now prove that $u$ is diagonal. In order to do so, we first show that $u\ket{0} = u_0 \ket{0}$ for some phase $u_0$.
As Eq. (\ref{Eq_Specdec}) is the spectral decomposition of $\rho_{n,d}^{(1,2)}$ and since $\alpha > \beta > \gamma$ (in fact, $\alpha=\beta+\gamma$, as the state is critical), the following equation must hold,
\begin{align}
 &u\otimes u \ket{00} = \phi_{0}\ket{00} + \phi_{d-1}\ket{d-1 d-1}, \label{eq:ES5_1}
\end{align}
for some coefficients $\phi_{0},\phi_{d-1}$.
Note that the local operator $u \otimes u$ cannot change the Schmidt rank of a state. Hence, the state on the right-hand side of Eq.  (\ref{eq:ES5_1}) must have Schmidt rank 1. That is, either $u\ket{0} = u_0\ket{0}$ or $u\ket{0} = u_{0} \ket{d-1}$ holds for some phase $u_0$. It is easy to see that only the former can fulfill Eq. (\ref{eq:EVn}). Hence, $u\ket{0} = u_0\ket{0}$ holds.\\

Next, we show that if $u\ket{k} = u_k\ket{k}$ for $k<i$ then also $u\ket{i}=u_i\ket{i}$ holds, for $i \in \{1,\ldots,d-1\}$. We consider the eigenspace to eigenvalue $\gamma$, which is spanned by the states $\{\ket{D_{1,2}(i)}\}_{i=1}^{d-1}$. Using that $u\ket{i-1} = u_{i-1}\ket{i-1}$ we obtain,
\begin{align*}
 u\otimes u \ket{D_{1,2}(i)} &= u_{i-1}(u\ket{i}\ket{i-1} +\ket{i-1}u\ket{i}),
\end{align*}
for $i \in \{1,\ldots,d-1\}$.
This state has to be an element of the eigenspace to eigenvalue $\gamma$. It is easy to see that this is only possible if $u\ket{i} = u_i \ket{i}$. Combined with $u\ket{0} = u_0\ket{0}$ we have that $u$ is diagonal.\\

We show next that for any diagonal $u$ fulfilling Eq. (\ref{eq:EVn}) it must hold that $u^{\otimes n} = \one$. Using the notation $u = \sum_{i=0}^{d-1} u_i \ket{i}\bra{i}$, it is straightforward to show that Eq. (\ref{eq:EVn}) is equivalent to
\begin{align}
& u_i^n = 1 \ \text{for} \ i \in \{0,\ldots,d-1\}, \label{eq:phases1}\\
& u_i^k u_{i-1}^{n-k} = 1 \ \text{for} \ i \in \{1,\ldots,d-1\} \label{eq:phases2}.
\end{align}
Note that $u$ is only uniquely determined up to a global phase $\omega_n^m$, where $\omega_n = \exp(2\pi i/n)$ and $m\in\mathbb{N}$. As $u_0^n = 1$ we can choose $u_0 = 1$ without loss of generality. Equations (\ref{eq:phases1}-\ref{eq:phases2}) then reduce to
\begin{align}
\label{eq:phases3}
 u_i^n = 1 \ \text{and} \ u_i^k u_{i-1}^{n-k} = 1 \ \text{for} \ i \in \{1,\ldots,d-1\}.
\end{align}
We now prove inductively that these equations together with $u_0 = 1$ imply that $u_i = 1$ for all $i$. That is, we show that $u_{i-1} = 1$ implies that $u_i=1$. Due to  Eq. (\ref{eq:phases3}) we have that $u_{i-1} = 1$ implies that $u_i^k = 1$. As $\text{gcd}(n,k) = 1$, a well known result form number theory implies that, there exist $a,b \in \mathbb{Z}$ such that $a n + b k = 1$. Hence, we have
\begin{align}
 u_i = u_i^{a n + b k} = 1.
\end{align}
We hence proved that $u_i = 1$ for all $i$, which implies $u^{\otimes n} = \one$. This proves the assertion. \qed
}\\

In the proof for $n>6$ (where $k < n-2$) we used Eq. (\ref{Eq_Specdec}). Let us now consider the case $n = 5$ (where $k = 3 = n-2$) for which we need to prove the statement differently.

\begin{lemma}
\label{lem:trivstabn5}
 The stabilizer of $\ket{\Psi_{5,d}}$ is trivial for $d\geq 2$, i.e. $\widetilde{G}_{\Psi_{n,d}} = \{\one\}$.
\end{lemma}

\proof{
Due to Lemma \ref{LemmaUnitStab} and Lemma \ref{lem:symSym4} it is again sufficient to show that any unitary $u$ that fulfills
\begin{align}
 u^{\otimes 5}\ket{\Psi_{5,d}} = \ket{\Psi_{5,d}} \label{eq:EV5}
\end{align}
fulfills $u^{\otimes 5} = \one$.\\

We first consider the necessary condition
 \begin{align}
  (u\otimes u) \rho_{5,d}^{(1,2)} (u^{\dagger} \otimes u^{\dagger}) = \rho_{5,d}^{(1,2)}, \label{eq:red1}
 \end{align}
where $\rho_{5,d}^{(1,2)} = \tr_{3,4,5}(\ket{\Psi_{5,d}}\bra{\Psi_{5,d}})$. Let $K_{d}$ denote the kernel of $\rho_{5,d}^{(1,2)}$ and let $K_{d}^{\bot}$ denote the orthogonal complement of $K_d$. Clearly, $K_{d}$ and $K_d^{\bot}$ have to be invariant under $u\otimes u$. In fact, $\ket{\psi} \in K_d$ iff $(u\otimes u)\ket{\psi} \in K_{d}$ (and similarly for $K_d^{\bot}$). In the following we use this to prove that $u$ is diagonal. Before that we have to characterize $K_d$ and $K_d^{\bot}$. It is straightforward to see that the following holds.
\begin{align}
 &K_d = Q \oplus S_{-}, \label{eq:kernel}\\
 &K_{d}^{\bot} = P \oplus S_{+}, \label{eq:orthkernel}
\end{align}
where
\begin{align}
&Q = \text{span}\{\ket{i}\ket{j}| 0\leq i,j \leq d-1, |i-j|>1\}, \label{eq:Q}\\
&S_{-} = \text{span}\{\ket{i}\ket{i-1}-\ket{i-1}\ket{i}\}_{i=1}^{d-1}.\\
&P = \text{span}\{\ket{i}\ket{i}\}_{i=0}^{d-1},\\
&S_{+} = \text{span}\{\ket{D_{1,2}(i)}\}_{i=1}^{d-1}. \label{eq:Splus}
\end{align}
Let us first prove that $u\ket{0} = u_0\ket{u}$ for some phase $u_0$.
Let $\pi^{\bot}$ denote the projector onto $K_d^{\bot}$, i.e.
\begin{align}
 \pi_{1,2}^{\bot} = \sum_{i=0}^{d-1} \ket{i}\bra{i}^{\otimes 2} + \frac{1}{2}\sum_{i=1}^{d-1} \ket{D_{1,2}(i)}\bra{D_{1,2}(i)}.
\end{align}
As $K_d^{\bot}$ is invariant under $u\otimes u$ it holds that
\begin{align}
 (u\otimes u)\pi_{1,2}^{\bot}(u^\dagger \otimes u^\dagger) = \pi_{1,2}^{\bot}. \label{eq:orthKinv}
\end{align}
For $d = 2$ it is easy to see that this equation and Eq. (\ref{eq:EV5}) can only be fulfilled if $u\ket{0} = u_0\ket{u}$ for some phase $u_0$.
For $d > 2$ let $\pi_{1}^{\bot} = \tr_2(\pi_{1,2}^{\bot})$ denote the reduced state of $\pi_{1,2}^{\bot}$, i.e.
\begin{align}
 \pi_1^{\bot} = \frac{3}{2}(\ket{0}\bra{0} + \ket{d-1}\bra{d-1}) + 2\sum_{i=1}^{d-2} \ket{i}\bra{i}.
\end{align}
Then Eq. (\ref{eq:orthKinv}) implies that
\begin{align}
 u \pi_1^{\bot} u^\dagger = \pi_1^{\bot}. \label{eq:local}
\end{align}
 It is easy to see that Eq. (\ref{eq:local}) can only be satisfied if $u\ket{0} \in \text{span}\{\ket{0},\ket{d-1}\}$. Combining this with the fact that $(u\otimes u)\ket{00} \in K_d^{\bot}$, we see that $(u\otimes u)\ket{00} \in \text{span}\{\ket{00}, \ket{d-1d-1}\}$.
However, as $u \otimes u$ cannot change the Schmidt rank of a state either $u\ket{0} = u_0 \ket{0}$ or
$u\ket{0} = u_0 \ket{d-1}$ for some phase $u_0$. It is easy to see that the former cannot fulfill Eq. (\ref{eq:EV5}) and thus $u\ket{0} = u_0\ket{0}$ holds.

Next, we show that if $u\ket{k} = u_{k}\ket{k}$ holds for $k < i$, then also  $u\ket{i} = u_{i}\ket{i}$ holds, where $i \in \{1,\ldots,d-1\}$. Using that
$u\ket{i-1} = u_{i-1}\ket{i-1}$ we get that
\begin{align}
(u\otimes u)\ket{D_{1,2}(i)} = u_{i-1}(u\ket{i}\ket{i-1} + \ket{i-1}u\ket{i}) \in K_d^{\bot}. \nonumber
\end{align}
The state on the right-hand side can only be an element of $K_d^{\bot}$ if $u\ket{i} = u_i \ket{i}$.
Combined with the fact that $u\ket{0} = u_0\ket{0}$ we see that $u$ is diagonal.\\

That $u^{\otimes 5} = \one$ holds for diagonal $u$ can then be proven in the same way as it was done in the proof of Lemma \ref{lem:trivstab7} for $n>6$. \qed
}\\

Note that this proof method could also be used to show that $\ket{\Psi_{n,d}}$ has trivial stabilizer for $n>6$. However, we think that the proof for $n>6$ presented after Lemma \ref{lem:trivstab7} is more concise.

\subsubsection{A critical $(n=4)$-qudit state with local dimension $d > 2$ and trivial stabilizer}
\label{sec:n4}

In this section we consider $4$-qudit systems with local dimension $d>2$. We define the (unnormalized) state
\begin{align}
 \ket{\Phi_{4,d}} &= \sqrt{15}c_0 \ket{0}^{\otimes 4} \nonumber \\
 &+ \sum_{i=1}^{d-1} c_i \left\{\ket{D_{3,4}(i)} + \ket{D_{2,4}(i)} -3 \ket{i}^{\otimes 4}\right\},
\end{align}
where
\begin{align}
 c_{i} = \frac{1}{3} \sqrt{1-\left(-\frac{4}{15}\right)^{d-i}}, \ \text{for} \ i \in \{0,\ldots,d-1\}. \label{eq:coeffs4}
\end{align}
It is easy to show that this state is critical. In what follows we prove that $\ket{\Phi_{4,d}}$ has only trivial symmetries.\\

Note that the following lemma holds.
\begin{lemma}
 \label{lem:symSym4}
 If $v \in \tilde{K}$ fulfills $v\ket{\Phi_{4,d}} = \ket{\Phi_{4,d}}$ for $d>2$ then $v = u^{\otimes 4}$ for some $u \in U(d)$.
\end{lemma}
 The proof is similar to the proof of Lemma \ref{lem:symSym5} and will be omitted. Using this lemma, we prove the following lemma.

\begin{lemma}
\label{lem:trivstab3}
For $d > 2$ the stabilizer of $\ket{\Phi_{4,d}}$ is trivial, i.e. $\widetilde{G}_{\Phi_{4,d}} = \{\one\}$.
\end{lemma}

\proof{

Due to Lemma \ref{LemmaUnitStab} and Lemma \ref{lem:symSym5} it remains to show that for any unitary $u\in U(d)$ the equation
\begin{align}
 u^{\otimes 4} \ket{\Phi_{4,d}} = \ket{\Phi_{4,d}} \label{eq:EV4}
\end{align}
implies that $u^{\otimes 4} = \one$. We first show that Eq. (\ref{eq:EV5}) implies that $u$ is diagonal and then that $u^{\otimes 4} = \one$.

Considering the reduced state of the first two subsystems, $\rho_{4,d}^{(1,2)}= \tr_{3,4}(\ket{\Phi_{4,d}}\bra{\Phi_{4,d}})$, Eq. (\ref{eq:EV4}) implies that
 \begin{align}
  (u\otimes u) \rho_{4,d}^{(1,2)} (u^{\dagger} \otimes u^{\dagger}) = \rho_{4,d}^{(1,2)}. \label{eq:red0}
 \end{align}

Let us denote by $K_d$ the kernel of $\rho_{4,d}^{(1,2)}$ and by $K_d^{\bot}$ the orthogonal complement of $K_d$. In Observation \ref{obs:appendA} we show that $K_d = Q \oplus S_{-}$ and $K_d^{\bot} = P \oplus S_{+}$, where $Q, S_-, P, S_+$ are given in Eqs. (\ref{eq:Q}-\ref{eq:Splus}). The proof that $u$ is diagonal can then be completed as in the proof of Lemma \ref{lem:trivstabn5} and will therefore be omitted.\\

Using that $u = \sum_{i=0}^{d-1} u_i \ket{i}\bra{i}$, Eq. (\ref{eq:EV4}) is equivalent to
\begin{align}
& u_i^4 = 1 \ \text{for} \ i \in \{0,\ldots,d-1\}, \label{eq:phases14}\\
& u_i^3 u_{i-1}^1 = 1 \ \text{for} \ i \in \{1,\ldots,d-1\} \label{eq:phases24}.\\
& u_i^2 u_{i-1}^2 = 1 \ \text{for} \ i \in \{1,\ldots,d-1\}.\label{eq:phases34}
\end{align}
Similar to the proof of Lemma \ref{lem:trivstab} we set, without loss of generality, $u_0 = 1$. It is then straightforward to see that Eqs. (\ref{eq:phases14}-\ref{eq:phases34}) only have the solution $u_i = 1$ for $i \in \{0,\ldots, d-1\}$. Hence, we proved that $u^{\otimes 4} = \one$ holds. \qed\\
}

 Here, we prove the following observation used in the proof of Lemma \ref{lem:trivstab3}.
 \begin{observation}
 \label{obs:appendA}
  Let $Q, S_-, P$ and $S_+$ be as defined in Eqs. (\ref{eq:Q} - \ref{eq:Splus}). Then the kernel of $\rho_{4,d}^{(1,2)}$ is $K_d = Q \oplus S_{-}$ and the orthogonal complement of the Kernel is $K_d^{\bot} = P \oplus S_{+}$.
 \end{observation}

\proof{
Note that $\ket{\psi} \in K_d$ iff
\begin{align}
 _{1,2}\langle\psi\ket{\Phi_{4,d}}_{1,2,3,4} = 0, \label{eq:inKernel}
\end{align}
where we explicitly labeled on which subsystems these states are define.
Using this, it is easy to see that $Q \oplus S_{-} \subset K_d$. As $\mathbb{C}^d \otimes \mathbb{C}^d = Q \oplus S_{-} \oplus P \oplus S_{+}$ it is therefore sufficient to show that $P \oplus S_{+}$ does not contain any nontrivial element of $K_d$, i.e. $(P \oplus S_{+}) \cap K_d = \{0\}$, in order to prove the observation.
That is, we have to show that an element $\ket{\psi} = \sum_{k=0}^{d-1} \alpha_k^* \ket{k}\ket{k} + \sum_{i=1}^{d-1}\beta_k^* \ket{D_{1,2}(k)} \in P \oplus S_{+}$ fulfills Eq. (\ref{eq:inKernel}) iff $\alpha_k, \beta_k = 0$ for all $k$.\\
First, it is easy to see that Eq. (\ref{eq:inKernel}) implies that $\beta_k = - \alpha_k/2$ for $k \in \{1,\ldots,d-1\}$. Using this and the notation $\vec{\alpha} = (\alpha_0,\ldots,\alpha_{d-1})^T$, it is straighforward to show that Eq. (\ref{eq:inKernel}) is equivalent to the system of linear equations
\begin{align}
 M \vec{\alpha} = 0, \label{eq:linsys}
\end{align}
where
\begin{widetext}
\begin{align*}
 M = \begin{pmatrix}
      \begin{matrix}
       \sqrt{15}c_{0} & c_{1} & 0 &0 & \ldots & 0 & 0 & 0 & 0\\
       c_{1} & -4c_{1} & c_{2} & 0 & 0 & \ldots & 0 & 0 & 0\\
       0 & c_{2} & -4c_{2} & c_{3} & 0 & \ldots & 0 & 0 & 0\\
       \vdots &  & \ddots & & & & \vdots & \vdots & \vdots\\
       0 & 0 & 0 & 0 & 0 & \ldots & c_{d-2} & -4c_{d-2} & c_{d-1}\\
       0 & 0 & 0 & 0 & 0 & \ldots & 0 & c_{d-1} & -4c_{d-1}
      \end{matrix}
     \end{pmatrix}.
\end{align*}
\end{widetext}
As $M$ is a tridiagonal matrix its determinant $\text{det}(M)$ can be computed via the following reccurence relation (see e.g. \cite{horn}).
\begin{align}
 &f(k) = -4c_{k-1}f(k-1) - c_{k-1}^2 f(k-2),\\
 &\text{for} \ k \in \{2,\ldots,d\}, \label{eq:sequence} \nonumber\\
 &f(1) = \sqrt{15} c_{0}, f(0) = 1,
\end{align}
where $\text{det}(M) = f(d)$.
Using the definition of $\{c_{k}\}_{k=0}^{d-1}$ one can show that $\{|f(k)|\}_{k=1}^{d}$ is monotonically increasing and hence $\text{det}(M) = f(d) \neq 0$ holds. That is, $M$ is invertible, $(P \oplus S_{+}) \cap K_d = \{0\}$ and therefore $K_d = Q \oplus S_{-}$, which proves the assertion. \qed\\
}

\subsubsection{A critical $(n=6)$-qudit state with trivial stabilizer for $d \geq 2$}
\label{sec:n6}

In this section we present a state in $(\C^d)^{\otimes 6}$, $d\geq 2$, with trivial stabilizer.
Before that, we introduce for $2 \leq j \leq d-1$ the (unnormalized) state
\begin{align}
 \ket{\phi_{6}(j)} = \ket{\pi(\ket{j}\ket{j-1}^{\otimes 3}\ket{j-2}^{\otimes 2})},
\end{align}
where $\ket{\pi(\ket{\psi})} = \sum_{\ket{\phi} \in \pi(\ket{\psi})} \ket{\phi}$ and $\pi(\ket{\psi}) = \{P_{\sigma}\ket{\psi} \ | \ \sigma \in S_n\}$ as in the main text.
For $d>3$ we then introduce the (unnormalized) critical state
\begin{align}
 \ket{\Phi_{6,d}} = \sqrt{\frac{194}{5}}\ket{0}^{\otimes 6} + \sqrt{\frac{11}{5}} \ket{D_{5,6}(1)} + \sum_{j=2}^{d-3} \ket{j}^{\otimes 6} \nonumber \\
 + \sqrt{21} \ket{d-2}^{\otimes 6} + \sqrt{51} \ket{d-1}^{\otimes 6} + \sum_{j=2}^{d-1} \ket{\phi_{6}(j)}. \label{eq:6qdstate}
\end{align}
Note that this state is not defined for $d = 2$ and not critical for $d=3$. However, for these case we can define the critical states,
\begin{align*}
 &\ket{\Phi_{6,2}} = 2 \ket{0}^{\otimes 6} + \ket{D_{5,6}(1)} + \ket{D_{3,6}(1)},\\
 &\ket{\Phi_{6,3}} = 3\ket{0}^{\otimes 6} + \ket{D_{5,6}(1)} + \frac{1}{\sqrt{2}}\ket{\phi_6(2)} + \sqrt{15} \ket{2}^{\otimes 6}.
\end{align*}
Note that, for $d>3$, the state $\ket{\Phi_{6,d}}$ can also be expressed as
\begin{widetext}
\begin{align}
 \ket{\Phi_{6,d}} &= c_0\ket{0}^{\otimes 2}\ket{0}^{\otimes 4} + c_1\ket{1}^{\otimes 2}\ket{D_{3,4}(1)} + c_1 \ket{D_{1,2}(1)}\ket{1}^{\otimes 4} + \sum_{j=2}^{d-3} \ket{j}^{\otimes 2}\ket{j}^{\otimes 4}+ \sum_{j=2}^{d-1} \{ \ket{j-2}^{\otimes 2}\ket{D_{1,4}(j)} \nonumber \\
 &+ \ket{j-1}^{\otimes 2}\ket{\alpha_j} + \ket{D_{1,2}(j)}\ket{D_{2,4}(j-1)} + (\ket{j}\ket{j-2}+\ket{j-2}\ket{j})\ket{D_{3,4}(j-1)} + \ket{D_{1,2}(j-1)}\ket{\beta_j} \}\nonumber\\
 &+ c_{d-2}\ket{d-2}^{\otimes 2}\ket{d-2}^{\otimes 4} + c_{d-1}\ket{d-1}^{\otimes 2}\ket{d-1}^{\otimes 4}. \label{eq:longdec}
\end{align}
\end{widetext}
Here, we defined $\ket{\alpha_j} = \ket{\pi(\ket{j}\ket{j-1}\ket{j-2}^{\otimes 2})}$,
$\ket{\beta_j} = \ket{\pi(\ket{j}\ket{j-1}^{\otimes 2}\ket{j-2})}$ and $c_0 = \sqrt{\frac{194}{5}}$,
$c_1 = \sqrt{\frac{11}{5}}, c_{d-2} = \sqrt{21}$ and $c_{d-1} = \sqrt{51}$. This decomposition will be convenient in later considerations.

Let us now show that $\ket{\Phi_{6,d}}$ has trivial stabilizer for $d>2$. For that, the following lemma is useful.

\begin{lemma}
 \label{lem:symSym6}
 If $v \in \tilde{K}$ fulfills $v\ket{\Phi_{6,d}} = \ket{\Phi_{6,d}}$ then $v = u^{\otimes 6}$ for some $u \in U(d)$.
\end{lemma}
The proof is similar to the proof of Lemma \ref{lem:symSym5} and will be omitted. We are now in the position to prove that $\ket{\Phi_{6,d}}$ has trivial stabilizer for $d > 2$.

\begin{lemma}
\label{lem:trivstab2}
For $d\geq 2$ the stabilizer of $\ket{\Phi_{6,d}}$ is trivial, i.e. $\widetilde{G}_{\Phi_{6,d}} = \{\one\}$.
\end{lemma}

\proof{
Due to Lemma \ref{LemmaUnitStab} and Lemma \ref{lem:symSym6} it is again sufficient to show that any unitary $u$ that fulfills
\begin{align}
 u^{\otimes 6}\ket{\Phi_{6,d}} = \ket{\Phi_{6,d}} \label{eq:EV6}
\end{align}
fulfills $u^{\otimes 6} = \one$. We divide the proof into three parts. In parts (a) and (b) we consider the cases with $d>3$. In part (c) and (d) we provide a proof of the theorem for $d=2$ and $d=3$, respectively.\\
\paragraph*{(a) $u$ is necessarily diagonal for $d>3$}\hfill\\
We consider the case $d>3$. From Eq. (\ref{eq:EV6}) we get the following necessary condition,
 \begin{align}
  (u\otimes u) \rho_{6,d}^{(1,2)} (u^{\dagger} \otimes u^{\dagger}) = \rho_{6,d}^{(1,2)}, \label{eq:redEV6}
 \end{align}
where $\rho_{6,d}^{(1,2)} = \tr_{3,\ldots,6}(\ket{\Phi_{6,d}}\bra{\Phi_{6,d}})$.
Let us first determine $\rho_{6,d}^{(1,2)}$.

With the help of decomposition (\ref{eq:longdec}) one can easily show that the $2$-subsystem reduced state of $\ket{\Phi_{6,d}}$ for $d>3$ reads
\begin{align}
 &\rho_{6,d}^{(1,2)} = \frac{214}{5}\ket{0}\bra{0}^{\otimes 2} +17 \sum_{j=2}^{d-3}\ket{j}\bra{j}^{\otimes 2}  \nonumber\\
 & + 33 \ket{d-2}\bra{d-2}^{\otimes 2} + 51 \ket{d-1}\bra{d-1}^{\otimes 2} \nonumber\\
 &+ \frac{71}{5} \ket{D_{1,2}(1)}\bra{D_{1,2}(1)}+18\sum_{j=2}^{d-2}\ket{D_{1,2}(j)}\bra{D_{1,2}(j)} \nonumber\\
 &+ 6\ket{D_{1,2}(d-1)}\bra{D_{1,2}(d-1)} +4\sum_{j=3}^{d-1} (\ket{j}\ket{j-2}\nonumber\\
 &+\ket{j-2}\ket{j})(\bra{j}\bra{j-2}+\bra{j-2}\bra{j}) \nonumber\\
 &+R, \label{eq:red6}
\end{align}
where
\begin{align*}
 R &= (4c_1^2+16) \ket{11}\bra{11}\\
 &+ 4 c_1 \left[\ket{11}(\bra{20}+\bra{02}) + c.c.\right] \\
 &+ 4 (\ket{20}+\ket{02})(\bra{20}+\bra{02}),
\end{align*}
where $c.c.$ denotes the complex conjugate.
Note that the only part of $\rho_{6,d}^{(1,2)}$ that is not yet diagonalized is the operator $R$. The only nonzero eigenvalues of $R$ are $\frac{2}{5}(41 \pm \sqrt{881})$. As a consequence, we can directly get a basis of all eigenspaces with eigenvalues different from  $\frac{2}{5}(41 \pm \sqrt{881}),0$ from Eq. (\ref{eq:red6}) .
For example, the states $\ket{0}^{\otimes 2}, \ket{d-2}^{\otimes 2} ,\ket{d-1}^{\otimes 2}$ span the 1-dimensional eigenspaces of $\rho_{6,d}^{(1,2)}$ with eigenvalues $\frac{214}{5}, 33, 51$, respectively. Note further that $u\otimes u$ maps an eigenspace of $\rho_{6,d}^{(1,2)}$ to a given eigenvalue to itself.
Combining these observations one can easily see that
\begin{align}
 &u \ket{0} = u_0 \ket{0}, \label{eq:1dES1}\\
 &u \ket{d-2} = u_{d-2} \ket{d-2},\label{eq:1dES2}\\
 &u \ket{d-1} = u_{d-1} \ket{d-1},\label{eq:1dES3}
\end{align}
for some phases $u_0, u_{d-1}, u_{d-2}$. For $d=4$ this already implies that $u$ is diagonal. For $d>4$ we can consider the eigenspaces to nonzero eigenvalues that have dimensions larger than $1$ to show that $u$ is diagonal.\\

Using that $u\ket{d-2} = u_{d-2} \ket{d-2}$ (see Eq.(\ref{eq:1dES2})) it remains to show that if $u\ket{k} = u_{k}\ket{k}$ holds for $k > i$, then also  $u\ket{i} = u_{i}\ket{i}$ holds, where $i \in \{2,\ldots,d-3\}$. Using that $u\ket{i+1} = u_{i+1}\ket{i+1}$ we get the following,
\begin{align}
 &(u\otimes u)\ket{D_{1,2}(i+1)} = u_{i+1} (\ket{i+1}u\ket{i} + u\ket{i}\ket{i+1}). \label{eq:sub18eq2}
\end{align}
Note that $u\otimes u$ maps the eigenspace of $\rho_{6,d}^{(1,2)}$ to eigenvalue $18$, spanned by $\{\ket{D_{1,2}(j)}\}_{j=2}^{d-2}$, to itself. Hence, Eq. (\ref{eq:sub18eq2}) can only be fulfilled if $u\ket{i} = u_i \ket{i}$ for some phase $u_i$. Combined with Eqs. (\ref{eq:1dES1}-\ref{eq:1dES3}) we have that $u\ket{i} = u_i\ket{i}$ for $i \in \{0\} \cup \{2,\ldots,d-1\}$. As $u$ is unitary, this also implies that $u\ket{1} = u_1\ket{1}$ for some phase $u_1$ and therefore $u$ is diagonal.\\

\paragraph*{(b) $u^{\otimes 6} = \one$ is the only solution for $d>3$}\hfill\\
Using that $u = \sum_{i=0}u_i\ket{i}\bra{i}$ for $d>3$ it is easy to see that Eq. (\ref{eq:EV6}) is equivalent to
\begin{align}
 &u_i^6 = 1 \ \text{for} \ i\in\{0,\ldots,d-1\} \label{eq:phases4}\\
 &u_0u_1^5 = 1\\
 &u_iu_{i-1}^3u_{i-2}^2 = 1 \ \text{for} \ i\in \{2,\ldots,d-1\}. \label{eq:phases5}
\end{align}
As in the proof of Lemma (\ref{lem:trivstab}) we can set, without loss of generality,
$u_0 = 1$. According to Eqs. (\ref{eq:phases4}-\ref{eq:phases5}) this implies $u_1^5 = 1$ and $u_2u_1^3=1$. Taking
the 6-th power of the second equation we obtain $u_2^6u_1^{18}= u_1^3 =1$. The only solution to $u_1^3 = u_1^5 = 1$ is $u_1 = 1$. Using that $u_0 = u_1 = 1$ in
Eq. (\ref{eq:phases5}) we finally obtain $u_i = 1$ for $i \in \{0,\ldots,d-1\}$. Hence, $u^{\otimes 6} = \one$ for $d>3$.\\

\paragraph*{(c) $u^{\otimes 6} = \one$ is the only solution for $d=2$}\hfill\\
Using that $u\otimes u$ leaves eigenspaces of $\rho_{6,2}^{(1,2)}$ invariant it is straigthforward to see that $u$ is diagonal. Reinserting this into Eq. (\ref{eq:EV6}) shows that indeed $u^{\otimes 6} = \one$.\\

\paragraph*{(d) $u^{\otimes 6} = \one$ is the only solution for $d=3$}\hfill\\
In this case it is easy to see that
\begin{align}
 \rho_{6,3}^{(1,2)} &= 11 \ket{0}\bra{0}^{\otimes 2} + 10 \ket{1}\bra{1}^{\otimes 2} + 15 \ket{2}\bra{2}^{\otimes 2} \nonumber \\
 &+ 7 \ket{D_{1,2}(1)}\bra{D_{1,2}(1)} + 3 \ket{D_{1,2}(2)}\bra{D_{1,2}(2)} \nonumber \\
 &+ 2(\ket{02}+\ket{20})(\bra{02}+\bra{20}) \nonumber\\
 & + 2\sqrt{2}[(\ket{02}+\ket{20})\bra{11}+ c.c.] \label{eq:red63}
\end{align}
Using that $u\otimes u$ maps eigenspaces of $\rho_{6,d}^{(1,2)}$ with a given eigenvalue to themselves it is easy to see that $u\ket{0} = u_0 \ket{0}$ and
$u\ket{2} = u_2\ket{2}$ for some phases $u_0,u_2$. As $u$ is unitary this implies that $u = \sum_{i=0}^2 u_i \ket{i}\bra{i}$. This form can then be reinserted into Eq. (\ref{eq:EV6}) to verify that $u^{\otimes 6} = \one$ is in fact the only solution of this equation. \qed
}

\subsection{A way to construct critical tripartite states with trivial stabilizer for $d>3$}
\label{sec:tripartite}

In this section we present a method to look for critical tripartite states with trivial stabilizer in $\tilde{G}$. It is particularily useful if one wants to find nonsymmetric states with these properties. We then employ this approach to explicitly construct such a state for local dimension $d=4,5,6$.

\subsubsection{On the symmetries of certain tripartite states}
Let $\ket{\psi} \in (\mathbb{C}^d)^{\otimes 3}$ be a critical state.
Due to Lemma \ref{LemmaUnitStab} we know that $\ket{\psi}$ does not have any nontrivial symmetries iff it does not have any nontrivial unitary symmetries. In what follows we derive necessary and sufficient conditions for some $\ket{\psi}$ that allow to determine when this is the case.
We consider the following critical states in $(\mathbb{C}^d)^{\otimes 3}$,
\begin{align}
 \ket{\psi} = \frac{1}{\sqrt{d}}\sum_{j=0}^{d-1} \ket{j}\otimes (U_j \otimes \one)\ket{\phi^+}, \label{eq:dec3}
\end{align}
where $\ket{\phi^+} = \frac{1}{\sqrt{d}}\sum_{i=0}^{d-1}\ket{ii}$, and the operators $\{U_i\}_{i=0}^{d-1}$ are unitaries that fulfill $\tr(U_i^\dagger U_j) = d\delta_{ij}$, i.e. that are orthogonal.
Suppose $V_1 \otimes V_2 \otimes V_3$ is a unitary symmetry, i.e.
\begin{align}
 V^{(1)} \otimes V^{(2)} \otimes V^{(3)}\ket{\psi} = \ket{\psi}. \label{eq:EV3}
\end{align}
We use the notation $V^{(1)} = \sum_{ij}V^{(1)}_{ij} \ket{i}\bra{j}$.
With the help of decomposition (\ref{eq:dec3}) and using that $\one \otimes A\ket{\phi^+} = A^T\otimes \one \ket{\phi^+}$ for all matrices $A$ it is then easy to see that
Equation (\ref{eq:EV3}) is equivalent to
\begin{align}
 \sum_j V^{(1)}_{ij} (V^{(2)}U_j{V^{(3)}}^T \otimes \one)\ket{\phi^+} = (U_i \otimes \one)\ket{\phi^+} \ \forall i. \label{eq:EV3_2}
\end{align}
Using that $(A \otimes \one)\ket{\phi^+} = (B \otimes \one)\ket{\phi^+}$ iff $A = B$, we can rewrite (\ref{eq:EV3_2}) as
\begin{align}
 \sum_j V^{(1)}_{ij} U_j = {V^{(2)}}^{\dagger}U_i(V^{(3)})^\ast \ \forall i. \label{eq:EV3_3}
\end{align}
These equations are equivalent to Eq. (\ref{eq:EV3}). Let us use the notation $W_i \equiv {V^{(2)}}^{\dagger}U_i(V^{(3)})^\ast$. Then a direct consequence of Eq. (\ref{eq:EV3_3}) is the following necessary condition,
\begin{align}
 W_iW_i^\dagger = \one + \sum_{j\neq k} V^{(1)}_{ij} {V^{(1)}_{ik}}^* U_j U_k^\dagger = \one \ \forall i.
\end{align}
That is,
\begin{align}
 \sum_{j\neq k} V^{(1)}_{ij} {V^{(1)}_{ik}}^* U_j U_k^\dagger = 0 \ \forall i \label{eq:EV3_4}
\end{align}
has to hold.\\

\subsubsection{Special critical states with trivial stabilizer}
From now on, we consider critical states for which Eq. (\ref{eq:EV3_3}) admits a particularily simple form. More precisely, we consider states for which $U_0 = \one$ and for which the unitaries $\{U_iU_j^\dagger\}_{i\neq j}$ are linearly independent. The second condition implies that Eq. (\ref{eq:EV3_4}) can only be fulfilled if $V^{(1)}_{ij} {V^{(1)}_{ik}}^* = 0$ for all $i$ and for all $j \neq k$, i.e. only if $V^{(1)}$ has exactly one nonzero entry in each row. That is, only if $V^{(1)}_{ij} = e^{i\phi_i} \delta_{i,\sigma(i)}$, where $\sigma \in S_d$ is a permutation. In the particular case we are considering, Eq. (\ref{eq:EV3_3}) therefore simplifies to
\begin{align}
 e^{i\phi_i} U_{\sigma(i)} = {V^{(2)}}^{\dagger}U_i{V^{(3)}}^* \ \forall i. \label{eq:EV3_5}
\end{align}
Due to Eq. (\ref{eq:EV3_5}) it holds that ${V^{(2)}}^{\dagger} = e^{i\phi_{\sigma^{-1}(0)}} {V^{(3)}}^T {U_{\sigma^{-1}(0)}}^\dagger$. Reinserting this into Eq. (\ref{eq:EV3_5}) and using the notation $\tilde{U} \equiv U_{\sigma^{-1}(0)}$ and $\tilde{\phi}_i \equiv \phi_{\sigma^{-1}(0)} - \phi_i$ we obtain
\begin{align}
 U_{\sigma(i)} = e^{i\tilde{\phi}_i} {V^{(3)}}^T \tilde{U}^\dagger U_i {V^{(3)}}^* \ \forall i. \label{eq:EV3_6}
\end{align}

In the following section we use the insights gained in this section to explain in detail how a tripartite state with local dimension $d=4$ and trivial stabilizer can be explicitly constructed. After that we show how the same method can be applied to $d=5$ and $d=6$.

\subsubsection{A critical tripartite state with local dimension $d=4$ and trivial stabilizer}
\label{sec:d4}

Let us construct a critical tripartite state with local dimension $d=4$ and trivial stabilizer. 
Our goal is to find unitaries $\{U_0,U_1,U_2,U_3\}$ such that Eq. (\ref{eq:EV3_5}) is only fulfilled for $\sigma = id$ and $V^{(3)} = e^{i\alpha_3} \one$ for some $\alpha \in \mathbb{R}$.  These conditions then imply that $V^{(1)} \otimes V^{(2)} \otimes V^{(3)} = \one$ as we show now.

Recall that $V^{(1)}_{ij} = e^{i\phi_i} \delta_{i,\sigma(i)}$. As $\sigma = id$ the unitary $V^{(1)}$ is a phase gate. Using that $V^{(3)} = e^{i\alpha_3} \one$ in Eq. (\ref{eq:EV3_5}) it is moreover easy to see that $V^{(2)} = e^{-i(\phi_i+\alpha_3)} \one$. Hence the phases $e^{i\phi_i}$ cannot depend on $i$ and  therefore fulfill $e^{i\phi_i} = e^{i\alpha_1}$ for some $\alpha_1 \in \mathbb{R}$. Thus, $V^{(1)} = e^{i\alpha_1} \one$. Hence, we have that $V^{(1)} \otimes V^{(2)} \otimes V^{(3)} = \one$.\\

In the following we show that the conditions $\sigma = id$ and $V^{(3)} = e^{i\alpha_3} \one$ are fulfilled if we choose
\begin{align}
&U_0 = \one, \label{eq:U0}\\
&U_1 = \frac{1}{10} T_1 \text{diag}(6+8i, -6-8i, -6+8i, 6-8i)T_1^\dagger, \\
&U_2 = \frac{1}{10^2} T_2 \text{diag}(96+28i, 96-28i, -96+28i, -96-28i)T_2^\dagger,\\
&U_3 = \frac{1}{10^3} \times \nonumber\\
&\text{diag}(936 + 352 i,-936 - 352 i,-936 + 352i, 936 - 352i), \label{eq:U3}
\end{align}
where
\begin{align*}
 &T_1 = \frac{1}{2}
 \begin{pmatrix} 1 & 1 & 1& 1 \\
 1 & -i & -1 & i \\
 -1 & 1 & -1 & 1 \\
 -1 & -i & 1 & i
 \end{pmatrix},\\
 &T_2 = \frac{1}{2}
 \begin{pmatrix} i & i & 1 & -1 \\
 \frac{1+i}{\sqrt{2}} & -\frac{1-i}{\sqrt{2}} & -\frac{1+i}{\sqrt{2}} & \frac{1-i}{\sqrt{2}} \\
 -i & i & 1 & 1 \\
 \frac{1+i}{\sqrt{2}} & \frac{1-i}{\sqrt{2}} & \frac{1+i}{\sqrt{2}} & \frac{1-i}{\sqrt{2}}
 \end{pmatrix}.
\end{align*}
Note that $T_1$ and $T_2$ transform the computational basis into the eigenbasis of the generalized Pauli operators $X_4Z_4^2$ and $X_4^3Z_4$, respectively, where
\begin{align*}
 &X_4 = \sum_{k = 0}^{3} \ket{k+1 \ \text{mod} \ 4}\bra{k},\\
 &Z_4 = \text{diag}(1,i,-1,-i).
\end{align*}
Hence, the unitaries in $\{U_i\}$ have the same eigenbases as the matrices in $\{\one, X_4Z_4^2, X_4^3Z_4, Z_4\}$.
However, their spectra are different. The choice $\{\one, X_4Z_4^2, X_4^3Z_4, Z_4\}$ would give rise to nontrivial solutions of Eq. (\ref{eq:EV3_5}) and therefore to nontrivial symmetries. One can show that the same happens if one choses any other subset of generalized Pauli operators to define the unitaries $\{U_i\}_{i=0}^3$.

It is straightforward to show that the unitaries $\{U_0, U_1, U_2, U_3\}$ fulfill the requirements necessary for Eq. (\ref{eq:EV3_6}) to be valid. That is, they are all mutually orthogonal and the unitaries  $\{U_iU_j^\dagger\}_{i\neq j}$ are linearly independent.
In what follows we show that the only choice of $\tilde{U} = U_{\sigma^{-1}(0)}$ that can fulfill Eq. (\ref{eq:EV3_6}) is $\tilde{U} = \one$. It is straightforward to see that Eq. (\ref{eq:EV3_6}) can only be fulfilled if for any $i \in \{0,1,2,3\}$ the spectrum of $\tilde{U}^\dagger U_i$ is proportional to the spectrum of $U_{\sigma(i)}$. It is however easy to see that this necessary condition cannot be fulfilled for $\tilde{U} \neq \one$. For example, for $\tilde{U} = U_1$ the spectrum of $\tilde{U}^\dagger U_3$ is not proportional to the spectrum of any of the unitaries in $\{U_i\}$. Consequently, $\tilde{U} = \one$ is the only way to satisfy Eq. (\ref{eq:EV3_6}), which then simplifies to
\be
U_{\sigma(i)} = e^{i\tilde{\phi}_i} {V^{(3)}}^T U_i {V^{(3)}}^* \ \forall i. \label{eq:EV3_7}
\ee
As the spectra of the $\{U_i\}$ are not proportional to each other, the only way to fulfill Eq. (\ref{eq:EV3_7}) is if $\sigma = id$. It is moreover easy to see that
$e^{i\tilde{\phi}_0} = 1$ and $e^{i\tilde{\phi}_j} \in \{-1,+1\}$ for $j \in \{1,2,3\}$.

As $e^{i\tilde{\phi}_1} \in \{-1,+1\}$ holds, we square Eq. (\ref{eq:EV3_7}) for $i=1$ and obtain
\begin{align}
&U_1^2 = {V^{(3)}}^T U_1^2 {V^{(3)}}^*. \label{eq:EV3_11}
\end{align}

Note that the spectrum of $U_1^2$ is degenerate as we have
\be
U_1^2 = -\frac{1}{25} T_1 \text{diag}(7-24i,7-24i,7+24i,7+24i)T_1^\dagger. \nonumber
\ee

Hence, Eq. (\ref{eq:EV3_11}) can only be fulfilled if ${V^{(3)}}^T$ is of the form
\begin{align}
 {V^{(3)}}^T = T_1 B T_1^\dagger, \label{eq:EV3_12}
\end{align}
where the unitary matrix $B$ is a block diagonal matrix, i.e. $B = \text{diag}(B_1,B_2)$, where $B_1,B_2$ are unitary $2 \times 2$ matrices.
Using the form of $V_3^T$ given in Eq. (\ref{eq:EV3_12}) it is easy to show that Eq. (\ref{eq:EV3_7}) with $e^{i\tilde{\phi}_j} \in \{-1,+1\}$ for all $j$ can only be fulfilled if $e^{i\tilde{\phi}_j} = 1$ for all $j$ and if ${V^{(3)}}^T = c\one$ for some $c \neq 0$.\\

In summary, we showed that $\sigma = id$ and $V^{(3)} = c\one$ for some $c \neq 0$. As outlined at the beginning of this section these conditions can then be used to show that indeed $V^{(1)} \otimes V^{(2)} \otimes V^{(3)} = \one$ is the only solution. Hence, the state corresponding to the unitaries $\{U_0, U_1, U_2, U_3\}$ in Eqs. (\ref{eq:U0} - \ref{eq:U3}), i.e. the state
\be
\ket{\psi} = \frac{1}{2}\sum_{j=0}^{4} \ket{j}\otimes (U_j \otimes \one)\ket{\phi^+},
\ee
has trivial stabilizer in $\tilde{G}$.

\subsubsection{Tripartite states with $d=5,6$ and trivial stabilizer}
\label{sec:3partynumerics}

In this section we use the techniques presented in the previous section to prove that there exist critical tripartite states with local dimension $d=5,6$ and trivial stabilizer.\\

As in the case of $d=4$ we explicitly construct unitaries $\{U_i\}_{i=0}^{d-1}$ for $d = 5,6$ that are diagonal in the eigenbasis of certain generalized Pauli operators. We then prove that the $\{U_i\}$ are orthogonal and show that the unitaries $\{U_iU_j^\dagger\}_{i \neq j}$ are linearly independent. This shows that Eq. (\ref{eq:EV3_5}) can be used to prove that the corresponding critical quantum state has trivial stabilizer. Analogously to the case of $d = 4$, we show that $\sigma = id$ and $V^{(3)} = c\one$ for some  $c \neq 0$ is the only solution of Eq. (\ref{eq:EV3_5}) for these choices of unitaries $\{U_i\}$. Since the matrices $\{U_iU_j^\dagger\}_{i \neq j}$ are linearly independent this then shows that the state
\be
\ket{\psi} = \frac{1}{\sqrt{d}}\sum_{j=0}^{d-1} \ket{j}\otimes (U_j \otimes \one)\ket{\phi^+}, \label{eq:trivstabnum}
\ee
is critical and has trivial stabilizer in $\tilde{G}$.\\

Let us now present these unitaries $\{U_i\}_{i=0}^{d-1}$ for $d=5,6$.
Recall that for every $d\geq 2$ and for any $k = (k_1,k_2) \in \{0,\ldots,d-1\}^2$ a generalized Pauli operator is defined as
\begin{align}
\label{eq:genpauli}
 S_{d,k} = X_d^{k_1} Z_d^{k_2},
\end{align}
where
\begin{align*}
 &X_d = \sum_{k = 0}^{d-1} \ket{k+1 \ \text{mod} \ d}\bra{k},\\
 &Z_d = \sum_{k = 0}^{d-1} \omega_d^k \ket{k}\bra{k},
\end{align*}
and $\omega_d = \exp(2\pi i/d)$.
As shown in \cite{bandyo} the matrix $U_{d,t}$ transforms the computational basis into the eigenbasis of $S_{d,(1,t)}$ for $t \in \{0,\ldots,d-1\}$, where
\begin{align}
 U_{d,t} = \frac{1}{\sqrt{d}}\sum_{i,j=0}^{d-1}(\omega_d^j)^{d-i}(\omega_d^{-t})^{\sum_{l=i}^{d-1}l} \ket{i}\bra{j}. \label{eq:Utrans}
\end{align}
The unitaries $\{U_i\}_{i=0}^{d-1}$ for which we show that the state in Eq. (\ref{eq:trivstabnum}) is critical and has trivial stabilizer in $\widetilde{G}$ are then the following.
\begin{align*}
&d=5:\\
&U_0 = \one,\\
&U_1 = U_{5,1} \text{diag}(e^{i\beta_1},e^{-i\beta_1},e^{i\alpha_1},e^{-i\alpha_1},-1)U_{5,1}^\dagger,\\
&U_2 = \text{diag}(-1, e^{i\beta_2},e^{-i\beta_2},e^{i\alpha_2},e^{-i\alpha_2}), \\
&U_3 = S_{5,(1,3)},\\
&U_4 = S_{5,(3,1)},
\end{align*}
with $\alpha_1 = \pi/3, \alpha_2 = \pi/6$ and $\beta_i = \arccos(1/2-\cos(\alpha_i))$ for $i = 1,2$.
\begin{align*}
&d=6:\\
&U_0 = \one,\\
&U_1 = \frac{1}{10}U_{6,0} \text{diag}(6+8i,6-8i,-6+8i,-6-8i,i,-i)U_{6,0}^\dagger,\\
&U_2 = \frac{1}{100}\text{diag}(96+28i,-96-28i,i,-i,96-28i,-96+28i), \\
&U_3 = S_{6,(1,1)},\\
&U_3 = S_{6,(2,3)},\\
&U_3 = S_{6,(4,2)}.
\end{align*}

It is easy to verify that for these choices of $\{U_i\}$ the unitaries $\{U_iU_j^\dagger\}_{i \neq j}$ are linearly independent. Following the same argument as in the previous section one can then show that $\ket{\psi}$ (given in Eq. (\ref{eq:trivstabnum})) is critical with trivial stabilizer.

\end{document}